\documentclass[runningheads,11pt]{llncs}
\usepackage[a4paper,margin=25mm]{geometry}
\usepackage{subfig}
\usepackage{tabularx}
\usepackage{array,multirow,graphicx}
\usepackage{float}
\usepackage{csquotes}
\usepackage{makecell}
\usepackage{color}
\usepackage{cite}
\usepackage{amssymb, extarrows}
\usepackage{amsmath}
\usepackage{amsfonts}
\usepackage{amscd}
\usepackage{enumerate}
\usepackage{url}
\usepackage{cancel}
\usepackage[normalem]{ulem}
\usepackage{algorithm, algorithmic}
\usepackage{upquote}
\usepackage{multicol}

\usepackage{slashbox}  

\usepackage{color}
\usepackage[export]{adjustbox}
\usepackage{parcolumns}

\usepackage{mwe}

\def\ZZ{\mathbb{Z}}

\def\RR{\mathbb{R}}

\def\cal{\mathcal}
\def\bf{\mathbf}
\def\calL{\mathcal{L}}

\def\Pr{\mathrm{Pr}}

\def\u{\bf{u}}

\def\e{\bf{e}}

\def\x{\bf{x}}

\def\L{\Lambda}
\def\Lp{\Lambda^{\perp}}
\def\b{\bf{b}}
\def\s{\bf{s}}

\usepackage{color}
\usepackage{listings} 
\usepackage{xspace}  
\usepackage[misc,geometry]{ifsym}

\lstdefinelanguage{Sage}[]{Python}
{morekeywords={True,False,sage,singular},
	sensitive=true}
\definecolor{dblackcolor}{rgb}{0.0,0.0,0.0}
\definecolor{dbluecolor}{rgb}{.01,.02,0.7}
\definecolor{dredcolor}{rgb}{0.8,0,0}
\definecolor{dgraycolor}{rgb}{0.30,0.3,0.30}

\newcommand{\dblack}{\color{dblackcolor}\textbf}

\renewcommand{\emph}[1]{{\dblack{#1}}}

\providecommand{\keywords}[1]{\textbf{\textit{Key words: }} #1}

\begin{document}
	\title{\textbf{ Puncturable Encryption:\\ A Generic Construction from Delegatable Fully Key-Homomorphic Encryption }}
\titlerunning{DFKHE-based Puncturable Encryption}

	\author{Willy Susilo\inst{1}, Dung Hoang Duong\inst{1},  Huy Quoc Le\inst{1,2}, Josef Pieprzyk\inst{2}}
	
	\authorrunning{W. Susilo, D. H. Duong, H. Q. Le and J. Pieprzyk}
	\institute{
		Institute of Cybersecurity and Cryptology, School of Computing and Information Technology, University of Wollongong, 
		Northfields Avenue, Wollongong NSW 2522, Australia.\\
		 \email{\{wsusilo,hduong\}@uow.edu.au},
		 	\email{qhl576@uowmail.edu.au}
	\and
		CSIRO Data61, Sydney, NSW  2015, Australia,\\
		and 
		Institute of Computer Science, Polish Academy of Sciences, Warsaw, Poland.
		\email{Josef.Pieprzyk@data61.csiro.au } 
	}

	\maketitle

\begin{abstract} 
	
Puncturable encryption (PE), proposed by Green and Miers at IEEE S\&P 2015, is a kind of public key encryption that  allows recipients to revoke individual messages by repeatedly updating decryption keys without communicating with senders. PE is an essential tool for constructing many interesting applications, such as asynchronous messaging systems,  forward-secret zero round-trip time protocols, public-key watermarking schemes and forward-secret proxy re-encryptions. This paper  revisits PEs  from the observation that the puncturing property can be implemented as efficiently computable functions. From this view, we propose a generic PE construction from the fully key-homomorphic encryption, augmented with a key delegation mechanism (DFKHE) from Boneh et al. at Eurocrypt 2014. We show that our PE construction enjoys the selective security under chosen plaintext attacks (that can be converted into the adaptive security with some efficiency loss)  from that of DFKHE in the standard model.  Basing on the framework, we obtain the first post-quantum secure PE instantiation that is based on the learning with errors problem, selective secure under chosen plaintext attacks (CPA) in the standard model. We also discuss about the ability of modification our framework to support the unbounded number of ciphertext tags inspired from the work of Brakerski and Vaikuntanathan at CRYPTO 2016.

\end{abstract}


\section{Introduction}

Puncturable encryption (PE), proposed by Green and Miers \cite{GM15} in 2015,  is a kind of public key encryption, 
which can also be seen as a tag-based encryption (TBE), where both encryption and decryption are controlled by tags. 
Similarly to TBE, a plaintext in PE is encrypted together with tags, which are called \textit{ciphertext tags}. 
In addition, the puncturing property of PE allows to produce new punctured secret keys associated 
some \textit{punctures} (or \textit{punctured tags}).  
Although the new keys (\textit{puncture keys}) differ from the old ones, they still allow recipients to decrypt old ciphertexts as long as chosen \textit{punctured tags} are different from tags embedded in the ciphertext. 
The puncturing  property is very useful when the current decryption key is compromised.
In a such situation, a recipient merely needs  to update his key using the puncturing mechanism.
PE is also useful when there is a need to revoke decryption capability from many users 
in order to protect some sensitive information (e.g., a time period or user identities).
In this case, the puncturing mechanism is called for time periods or user identities.

 Also, PE can provide  forward security in a fine-grained level.
Forward security, formulated in \cite{Gun90} in the context of key-exchange protocols, is a desired security property 
that helps to reduce a security risk caused by key exposure attacks.
In particular, forward secure encryption (FSE) guarantees confidentiality of old messages, 
when the current secret key has been compromised. 
Compared to PE, FSE provides a limited support for revocation of decryption capability.
For instance, it is difficult for FSE to control decryption capability for any individual ciphertext
(or all ciphertexts) produced during a certain time period, 
which, in contrast, can be easily done with PE.

Due to the aforementioned advantages, 
PE has become more and more popular and has been used in many important applications in such as
asynchronous messaging transport systems \cite{GM15},
forward-secure zero round--trip time  (0-RTT) key-exchange protocols \cite{GHJ+17, DJSS18}, 
public-key watermarking schemes \cite{CHN+16} and
forward-secure proxy re-encryptions  \cite{DKL+18}.\\

\noindent \textbf{Related Works.} Green and Miers \cite{GM15}  
propose the notion of  PE  and also present a specific ABE-based PE instantiation. 
The instantiation is based on the decisional bilinear Diffie-Hellman assumption (DBDH) in bilinear groups 
and is proven to be CPA secure in the random oracle model (ROM). 
Following the work \cite{GM15}, many other constructions have been proposed such as \cite{CHN+16, CRRV17, GHJ+17,  DJSS18, SSS+20}
(see Table \ref{tab2} for a summary). For instance, 
G{\"u}nther et al. \cite{GHJ+17} have provided a generic PE construction
from \textit{any} selectively secure hierarchical identity-based key encapsulation (HIBEKEM) 
combined with an \textit{any} one time signature (OTS). 
In fact, the authors of \cite{GHJ+17} claim that their framework can be instantiated as the first post-quantum PE. 
Also, in the work \cite{GHJ+17}, the authors present the first PE-based forward-secret zero round-trip time protocol with full forward secrecy. 
However, they instantiate PE that is secure in the standard model (SDM) by 
combining a (DDH)-based HIBE with a OTS based on discrete logarithm.
The construction supports a predetermined  number of ciphertext tags as well as a limited number of punctures.
Derler et al. \cite{DJSS18} introduce the notion of Bloom filter encryption (BFE),
 which can be converted to PE.
They show how to instantiate
BFE using identity-based encryption (IBE) with a 
specific construction that assumes intractability of
the bilinear computational
Diffie-Hellman (BCDH) problem. 
Later, 
Derler et. al. \cite{DGJ+18-ePrint} extend the result of \cite{DJSS18} 
and give a generic BFE construction from identity-based broadcast encryption (IBBE). 
The instantiation in \cite{DGJ+18-ePrint} is based on a generalization of the Diffie-Hellman exponent (GDDHE) assumption in parings. 
However, the construction based on BFE suffers from \textit{non-negligible correctness error}. 
This excludes it from applications that require 
negligible correctness error, as discussed in \cite{SSS+20}.
Most recently, Sun et al. \cite{SSS+20} have introduced a new concept, 
which they call key-homomorphic identity-based revocable key encapsulation mechanism (KH-IRKEM) with extended correctness, 
from which 
they obtain a modular design of PE with \textit{negligible correctness errors}.
In particular, they describe four modular and compact instantiations of PE, which are secure in SDM.
However, all of them are  based on  hard problems in pairings, 
namely $q$-decision bilinear Diffie-Hellman exponent problem ($q$--DBDHE), 
the decision bilinear Diffie-Hellman problem(DBDH), 
the $q$-decisional multi-exponent
 bilinear Diffie-Hellman (q-MEBDH) problem and  
 the decisional linear problem (DLIN). 
 We emphasize that all existing instantiations mentioned above are insecure against quantum adversaries. 
Some other works like \cite{CRRV17, CHN+16}  based PE on the notion of indistinguishability obfuscation, which is still impractical. 
The reader is referred to \cite{SSS+20} for a state-of-the-art discussion. 

To the best of our knowledge, there has been no specific lattice-based PE instantiation, which simultaneously enjoys  negligible correctness error as well as post-quantum security in the standard model.

 
\subsubsection{Our Contribution.}  
We first give a \textit{generic} construction of PE from \textit{delegatable fully key-homomorphic encryption} (DFKHE) framework. 
The framework is a generalisation of 
fully key-homomorphic encryption (FKHE) \cite{BGG+14} by adding
a key delegation mechanism.
The framework is closely related to the functional encryption \cite{BSW11}. 

We also present an  explicit  PE construction based on lattices.
Our design is obtained from LWE-based  DFKHE that we build using
FKHE for the learning with errors (LWE) setting \cite{BGG+14}. 
This is combined with the key delegation ability supplied by the lattice trapdoor techniques \cite{GPV08, CHKP10, ABB10}.  
Our lattice FE construction has the following characteristics:
\begin{itemize}
	\item It supports a predetermined number of ciphertext tags per ciphertext. 
		The ciphertext size is short and depends linearly on the number of ciphertext tags, 
		which is fixed in advance. However, we note that following the work of Brakerski and Vaikuntanathan \cite{BV16}, 
		our construction  might be extended to obtain a variant that supports unbounded number of 
		ciphertext tags (see Section \ref{unbounded} for a detailed discussion),
	\item
It works for a predetermined number of punctures.
		The size of decryption keys (i.e., puncture keys) increases quadratically with the number of punctured tags,
	\item  It offers selective CPA security in the standard model (that can be converted into full CPA security 
		using the complexity leveraging technique as discussed in \cite{CHK04}, \cite{Kil06}\cite{BB11}, \cite{BGG+14}).  
		This is due to CPA security of LWE-based underlying DFKHE (following the security proof for the generic framework).
	\item It enjoys post-quantum security and negligible correctness errors.
\end{itemize}

Table \ref{tab2} compares our work with the results obtained by other authors.
At first sight, the FE framework based on key homomorphic revocable identity-based (KH-IRKEM) \cite{SSS+20} 
looks similar to ours. 
However, both frameworks are different.
While key-homomorphism used by us means the capacity of transforming 
(as claimed in \cite[Subsection 1.1]{BGG+14}) 
``\textit{an encryption under key $\mathbf{x}$ into an encryption under key $f(\mathbf{x})$}", 
key-homomorphism  defined in \cite[Definition 8]{SSS+20}  
reflects the ability of  preserving the algebraic structure of (mathematical) groups. \\

\noindent \textbf{Overview and Techniques.}  
We start with a high-level description of fully-key homomorphism encryption (FHKE), 
which was proposed by Boneh et al. \cite{BGG+14}. 
Afterwards, we introduce what we call the \textit{delegetable  fully-key homomorphism encryption (DFHKE)}. 
At high-level description, FKHE possesses a mechanism that allows to convert a ciphertext $ct_\mathbf{x}$ 
(associated with a public variable
 $\mathbf{x}$) into the evaluated one $ct_f$ for the same plaintext (associated with the pair $(y,f)$),
 where $f$ is an efficiently computable function and $f(\mathbf{x})=y$. 
 In other words,
 FKHE requires a special key-homomorphic evaluation algorithm, 
 called $\mathsf{Eval}$, such that
  $ct_f \leftarrow \mathsf{Eval}(f, ct_\mathbf{x})$.
In order to successfully decrypt an evaluated ciphertext, 
the decryptor needs to evaluate the initial secret $sk$ to get $sk_f$.
An extra algorithm, called $\mathsf{KHom}$, is needed to do this, i.e.
$sk_f \leftarrow \mathsf{KHom}(sk,(y,f))$.  
A drawback of FKHE is that it supports only a single function $f$.  

Actually, we'd  like to perform
key-homomorphic evaluation for many functions $\{f_1,\cdots, f_k\}$ that belong to a family $\mathcal{F}$. 
To meet the requirement and obtain DFKHE,
we generalise FKHE  by endowing  it with two algorithms $\mathsf{ExtEval}$ and $\mathsf{KDel}$.
The first algorithm transforms 
 ($ct_\mathbf{x}, \mathbf{x}$) into ($ct_{f_1,\cdots, f_k}, (y,f_1, \cdots f_k)$), where $f_1(\mathbf{x})=\cdots=f_k(\mathbf{x})=y$. 
 This is written as $ct_{f_1,\cdots, f_k} \leftarrow \mathsf{ExtEval}(f_1,\cdots, f_k, ct_\mathbf{x})$.
 The second algorithm 
 allows to delegate the secret key step by step
 for the next function or $sk_{f_1, \cdots, f_k} \leftarrow \mathsf{KDel}(sk_{f_1, \cdots, f_{k-1}}, (y,f_k))$.

 \begin{table}[h]
 	\centering
 	\medskip
 	\smallskip
 	\small\addtolength{\tabcolsep}{0pt}
 	\begin{tabular}{ c | c| c|c|c|c |c|c}
 	\hline
 		Literature&From& Assumption&\makecell{Security\\ Model} &\#Tags &\#Punctures &\makecell{Post-\\quantum} &\makecell{Negl.\\ Corr. \\Error} \\
 		\hline
 		\hline 
 	Green	\cite{GM15}&ABE& DBDH& ROM &$<\infty$ &$\infty$ & $\times$ &$\checkmark$\\
 		\hline
 	G{\"u}nther 	\cite{GHJ+17} &\makecell{any HIBE \\+ any OTS}&\makecell{DDH (HIBE)\\+ DLP (OTS) }&SDM & $<\infty$ &$<\infty$ &$\times$ &$\checkmark$\\
 		\hline
 			Derler  \cite{DGJ+18-ePrint} &BFE (IBBE)& GDDHE& \textbf{ROM$^*$} &1& $< \infty$&$\times$ &$\times$\\
 		\hline
 		Derler \cite{DJSS18} &BFE (IBE)&BCDH&ROM &1& $<\infty$&$\times$&$\times$\\
 		\hline
 		Sun\cite{SSS+20}  &KH-IRKEM&\makecell{$q$--DBDHE\\DBDH\\ $q$--MEBDH\\ DLIN} &SDM& \makecell{$<\infty$\\$<\infty$\\$\infty$\\$<\infty$}&\makecell{$\infty$\\$\infty$\\$\infty$\\$\infty$} &\makecell{$\times$\\$\times$\\$\times$\\$\times$}&$\checkmark$\\
 		\hline
 		\hline
 		\textbf{This work}&	DFKHE &  DLWE& SDM& \makecell{$< \infty$}&$< \infty$&$\checkmark$ &$\checkmark$\\
 		\hline \hline
 	\end{tabular} 
 	
 	\caption{ Comparison of some existing PE constructions in the literature with ours.  Note that, here all works are being considered in the CPA security setting. The notation  ``$<\infty$" means "bounded" or ``predetermined", while ``$\infty $" means  ``unlimited" or ``arbitrary".  The column entitled ``Post-quantum" says whether the specific construction in each framework is post-quantum secure or not regardless its generic framework. The last column mentions to supporting the negligible correctness error. \textbf{ROM$^*$}: For the BFE-based  FE basing on the IBBE instantiation of Derler et al. \cite{DGJ+18-ePrint}, we note that, the IBBE instantiation can be modified to remove ROM, as claimed by Delerabl\'ee in \cite[Subection 3.2]{Del07} }
 	\label{tab2}
 \end{table}

 Our generic PE framework is inspired by a simple but subtle observation that puncturing property
 requires equality of ciphertext tags and punctures.
This can be provided by functions that can be efficiently computed by arithmetic circuits.
 We call such functions \textit{equality test functions}. 
Note that for PE, ciphertext tags play the role of variables $\textbf{x}$'s and equality test functions act as functions $f$'s  
defined in FKHE. 
For FE, one more puncture added defines one extra equality test function,
which needs a delegation mechanism to take the function into account. 
We note that the requirement can be 
easily met using the same idea as the key delegation mentioned above.
In order to be able to employ the idea of DFKHE for $(y_0, \mathcal{F})$ to PE, 
we define an efficiently computable family $\mathcal{F}$ of equality test functions $f_{t^*}(\mathbf{t})$ allowing us to compare the puncture $t^*$ with ciphertext tags $\mathbf{t}=(t_1, \cdots, t_d)$ under the definition that $f_{t^*}(\mathbf{t})=y_0$ iff $t^* \neq t_j \forall j\in [d]$, for some fixed value $y_0$. 

For concrete  DHKHE and PE constructions, we employ the LWE-based FKHE proposed in \cite{BGG+14}. 
 In this system, the ciphertext is $ct= (\textbf{c}_{\textsf{in}}, \textbf{c}_1, \cdots, \textbf{c}_d,  \textbf{c}_{\textsf{out}})$, 
where  $\textbf{c}_{i}=(t_i\textbf{G}+\textbf{B}_i)^T \textbf{s}+\textbf{e}_{i}$ for $i\in [d]$.
 Here the gadget matrix $\textbf{G}$ is a special one,
whose associated trapdoor $\textbf{T}_\textbf{G}$ 
(i.e., a short basis  for the $q$-ary lattice $\Lambda_q^{\bot}(\textbf{G})$)  
is publicly known (see \cite{MP12} for details).  
Also, there exist three evaluation algorithms named 
$\textsf{Eval}_{\textsf{pk}}$, $\textsf{Eval}_{\textsf{ct}}$ and $\textsf{Eval}_{\textsf{sim}}$ \cite{BGG+14},
which help us to homomorphically evaluate a circuit (function) for a ciphertext $ct$.  
More specifically, from $\textbf{c}_i:=[t_i\textbf{G}+\textbf{B}_i]^T\textbf{s}+\textbf{e}_i, \text{ where } \| \textbf{e}_i \|<\delta $ 
for all  $i\in [d]$, and a function $f:(\mathbb{Z}_q)^d \rightarrow \mathbb{Z}_q$, 
we get $\textbf{c}_f=[f(t_1, \cdots, t_d)\textbf{G}+\textbf{B}_f]^T\textbf{s}+\textbf{e}_f, \| \textbf{e}_f \|<\Delta $,
where
$ \mathbf{B}_f \leftarrow\mathsf{Eval}_\mathsf{pk}(f, (\mathbf{B}_i )_{i=1}^d) $, 
$  \mathbf{c}_f \leftarrow \mathsf{Eval}_\mathsf{ct}(f, ((t_i, \mathbf{B}_i,\mathbf{c}_i))_{i=1}^d)$,  
and $\Delta < \delta \cdot \beta$ for some $\beta$ sufficiently small. 
The algorithm $\textsf{ExtEval}$ mentioned above can be implemented calling many times $\mathsf{Eval}_\mathsf{pk}$,$\textsf{Eval}_{\textsf{ct}}$, each time for each function. Meanwhile,
$\mathsf{Eval}_\mathsf{sim}$ is only useful in the simulation for the security proof.  In the LWE-based DFKHE construction, secret keys are trapdoors for $q$-ary lattices of form $\Lambda_q^{\bot}([\textbf{A}|\textbf{B}_{f_1}|\cdots| \textbf{B}_{f_k}])$.  For the key delegation $\mathsf{KDel}$, we can utilize the trapdoor techniques \cite{GPV08, ABB10, CHKP10} .
For the LWE-based PE instantiation, we employ the equality test function with $y_0:=0 \text{ (mod } q)$.  Namely, for a puncture $t^*$ and a list of ciphertext tags $t_1, \cdots, t_d$ we define $f_{t^*}(t_1, \cdots, t_d):=eq_{t^*}(t_1)+\cdots+ eq_{t^*}(t_d)$, where $eq_{t^*}: \mathbb{Z}_q\rightarrow \mathbb{Z}_q $ satisfying that $\forall t\in \mathbb{Z}_q$, $eq_{t^*}(t)=1 \text{ (mod } q)$ iff $t=t^*$,
otherwise $eq_{t^*}(t)=0 \text{ (mod } q)$. Such functions has also been employed in \cite{BKM17} to construct a privately puncturable pseudorandom function. It follows from generic construction that our PE instantiation is selective CPA-secure.\\

\noindent \noindent \textbf{Efficiency.} Table \ref{tab21} summarizes the asymptotic bit-size of public key, secret key, 
punctured key and ciphertext. 
We can see that the public key size is a linear function in the number of ciphertext tags (i.e., $d$). 
The (initial) secret key size is independent of both $d$ and  $\eta$ (the number of punctures).  
The punctured key (decryption key) size  is a quadratic function of $\eta$. 
Lastly, the ciphertext size is a linear function of $d$.\\
	
\noindent \noindent{\textbf{On  unbounded ciphertext tags.}} 
We believe that our framework can be extended to support unbounded number of ciphertext tags 
by exploiting the interesting technique of \cite{BV16}. 
The key idea of \cite{BV16} is to use homomorphic evaluation of a family pseudorandom functions.
This helps to stretch a predetermined parameter (e.g., the length of a seed) to an arbitrary number of ciphertext tags. 
The predetermined parameter will be used to generate other public parameters (e.g., public matrices).  
More details is given in Section \ref{unbounded}.\\

\noindent \noindent{\textbf{Paper Organization.}}  In Section \ref{pre}, we review some background related to this work. Our main contributions are presented in Section  \ref{generic} and Section \ref{instan}. We formally define DFKHE and the generic PE construction from DFKHE in Section \ref{generic}. Section \ref{instan} is dedicated to the LWE-based instantiation of DFKHE and the induced PE. Section \ref{unbounded} discusses on the feasibility of transforming our proposed LWE-based PE to work well with  unbounded ciphertext tags. This work is concluded in Section \ref{conclude}.

\begin{table}[pt]
	\centering
	\medskip
	\smallskip
	\small\addtolength{\tabcolsep}{0pt}
	\begin{tabular}{ c|  c}
		\hline
		Public key size& $O((d+1)\cdot n^2 \log^2 q)$\\
		Secret key size&  $O(n^2 \log^2 q\cdot \log( n\log q))$  \\
		Punctured key size& $(\eta+1) \cdot n \log q \cdot( O(\log(\beta_{\mathcal{F}})+\eta\cdot \log (n \log q)))$\\
		Ciphertext size& $O((d+2)\cdot n\log^2q))$\\
		\hline
	\end{tabular} 
	
	\caption{  Keys and ciphertext's size of our LWE-based PE as functions in number of ciphertext tags $d$ and number of punctures $\eta$. }
	\label{tab21}
\end{table}

\section{Preliminaries} \label{pre}
\subsection{Framework of Puncturable Encryption } \label{peer}

\noindent \textbf{Syntax of puncturable encryption.} 
For a security parameter $\lambda$, let $d=d(\lambda)$, $\mathcal{M}=\mathcal{M}(\lambda)$ and $\mathcal{T}=\mathcal{T}(\lambda)$ be maximum number of tags per ciphertext,  the space of plaintexts and the set of valid tags, respectively.  
Puncturable encryption (PE)  is a collection of the following four algorithms 
\textsf{KeyGen}, \textsf{Encrypt}, \textsf{Puncture} and \textsf{Decrypt}:
\begin{itemize}
	\item \underline{$(pk, sk_0) \leftarrow \textsf{KeyGen}(1^\lambda, d)$: }
	For a security parameter $\lambda$ and
	the maximum number $d$ of tags per ciphertext, 
	the probabilistic polynomial time (PPT) algorithm \textsf{KeyGen} outputs a public key $pk$ and an initial secret key $sk_0$. 
	
	\item  \underline{ $ct \leftarrow \textsf{Encrypt}(pk,\mu, \{t_1, \cdots, t_d\})$: }
	For a public key $pk$, a message $\mu$, and a list of tags $t_1, \cdots, t_d$, 
	the PPT algorithm \textsf{Encrypt} returns a ciphertext $ct$.  
	
	\item  \underline{$sk_{i} \leftarrow \textsf{Puncture}(pk, sk_{i-1}, t^*_{i})$: }
	For any $i>1$, 
	on input $pk$, $sk_{i-1}$ and a tag $t^*_i$, the PPT algorithm \textsf{Puncture} outputs a punctured key $sk_{i}$ 
	that decrypts any ciphertexts, except for the ciphertext encrypted under any list of tags containing $t^*_i$.
	\item  \underline{ $\mu/\bot \leftarrow \textsf{Decrypt}(pk, sk_{i},(ct, \{t_1, \cdots, t_d\}))$: }
	For input $pk$, a ciphertext $ct$, a secret key $sk_{i}$, and a list of tags $\{t_1, \cdots, t_d\}$, 
	the deterministic polynomial time (DPT) algorithm \textsf{Decrypt} outputs either a message $\mu$ if the decryption succeeds or $\bot$ if it fails.
\end{itemize}

\noindent \textbf{Correctness.} The correctness requirement for PE is as follows: \\
For all  $\lambda, d, \eta\geq 0$, $t^*_1, \cdots, t^*_\eta ,t_1, \cdots, t_d \in \mathcal{T}$,   $(pk, sk_0) \leftarrow \textsf{KeyGen}(1^\lambda, d)$, $sk_i\gets\textsf{Punc}(pk, sk_{i-1},t^*_i),$ $ \forall i \in  [\eta]$, $ct=\textsf{Encrypt}(pk,\mu, \{t_1, \cdots, t_d\}),$ we have 
\begin{itemize}
\item If $\{t^*_1, \cdots, t^*_\eta\}\cap \{t_1, \cdots, t_d\}=\emptyset $, then  $\forall i\in \{0,\cdots, \eta\}$,
\[
\Pr[\textsf{Decrypt}(pk, sk_i, (ct, \{t_1, \cdots, t_d\}))=\mu]\geq 1-\textsf{negl}(\lambda).
\] 
\item If there exist $j\in[d]$ and $k\in [\eta]$ such that $t^*_k=  t_j$,  then $\forall i \in \{k,\cdots, \eta\}$,
\[
\Pr[\textsf{Decrypt}(pk, sk_i, (ct, \{t_1, \cdots, t_d\}))=\mu]\leq \textsf{negl}(\lambda).
\] 
\end{itemize}

\begin{definition}[Selective Security of PE]
	PE is IND-sPUN-ATK if the advantage of any PPT adversary $\mathcal{A}$ in 
	the game $\mathsf{IND}$-$\mathsf{sPUN}$-$\mathsf{ATK}^{\mathsf{sel},\mathcal{A}}_{\mathsf{PE}}$ is negligible,
	where ATK $\in $ \{CPA, CCA\}. 
	Formally,
	$$\mathsf{Adv}_{\mathsf{PE}}^{\mathsf{IND}\text{-}\mathsf{sPUN}\text{-}\mathsf{ATK}}(\mathcal{A})=|\Pr[b'=b]-\frac{1}{2}| \leq \mathsf{negl}(\lambda).$$
\end{definition}
The game $\mathsf{IND}$-$\mathsf{sPUN}$-$\mathsf{ATK}^{\mathsf{sel},\mathcal{A}}_{\mathsf{PE}}$ proceeds as follows.

\begin{enumerate}
	\item \textbf{Initialize.} 
	The adversary announces the target tags $\{\widehat{t_1}, \cdots, \widehat{t_d}\}$. 
	\item \textbf{Setup.} 
	The challenger initializes a set punctured tags $\mathcal{T}^* \leftarrow \emptyset$, 
	a counter  $i\leftarrow 0$ that counts the current number of punctured tags in $\mathcal{T}^*$ 
	and a set of corrupted tags $\mathcal{C}^* \leftarrow \emptyset$ containing all punctured tags at the time of the first corruption query. 
	Then, it runs $(pk, sk_0) \leftarrow \textsf{KeyGen}(1^\lambda, d)$. Finally, it gives $pk$ to the adversary.
	\item \textbf{Query 1.}
	\begin{itemize}
		\item Once the adversary makes a puncture key query PQ($t^*$), 
		the challenger updates $i \leftarrow i+1$, 
		returns $sk_{i} \leftarrow \textsf{Punc}(pk, sk_{i-1},t^*)$ and adds $t^*$ to $\mathcal{T}^*$.
		\item The first time the adversary makes a corruption query CQ(), 
		the challenger returns $\bot$ if it finds out that $\{\widehat{t_1}, \cdots, \widehat{t_d}\} \cap \mathcal{T}^*= \emptyset$. 		
		Otherwise, the challenger returns the most recent punctured key 
		$sk_\eta$, then sets $\mathcal{C}^*$ as  the most recent 
		$\mathcal{T}^*$ (i.e., $\mathcal{C}^* \leftarrow \mathcal{T}^*=\{t_1^*, \cdots, t^*_\eta\}$). 
		All subsequent puncture key queries and corruption queries are answered with $\bot$. 
		\item If $ATK=CCA$: Once the adversary makes a decryption query  DQ$(ct,\{t_1, $ $\cdots, t_d\})$, 
		the challenger runs $\textsf{Decrypt}(pk, sk_{\eta},  (ct,\{t_1, \cdots, t_d\}))$ 
		using the most recent punctured key $sk_\eta$ and returns its output. \\
		If $ATK=CPA$: the challenger returns $\bot$.
	\end{itemize}
	\item \textbf{Challenge.} 
	The adversary submits two messages $\mu_0, \mu_1$. 
	The challenger rejects the challenge if it finds out
	that $\{\widehat{t_1}, \cdots, \widehat{t_d}\} \cap \mathcal{C}^*= \emptyset$\footnote{Note that, after making some queries that are different from the target tags, the adversary may skip making corruption query but goes directly to the challenge phase and trivially wins the game. This rejection prevents the adversary from such a trivial win. It also force the adversary to make the corruption query before challenging.}. 
	Otherwise, the challenger chooses $b \xleftarrow{\$} \{0,1\}$ and 
	returns $\widehat{ct} \leftarrow \textsf{Encrypt}(pk,\mu_b, \{\widehat{t_1}, \cdots, \widehat{t_d}\})$.  
	\item \textbf{Query 2.} The same as Query 1 with the  restriction that
		for DQ$(ct,\{t_1, \cdots, t_d\})$, the challenger returns $\bot$ if $(ct,\{t_1, \cdots, t_d \})$ 
		$=(\widehat{ct}, \{\widehat{t_1}, \cdots, \widehat{t_d}\} )$.
	
	\item \textbf{Guess.} The adversary outputs $b'\in \{0,1\}$. It wins if $b'=b$. 
\end{enumerate}
 The full security for PE is defined in the same way, 
except that the adversary can choose target tags at Challenge phase,  
after getting the public key and after Query 1 phase. 
In this case, the challenger does not need to check the condition $\{\widehat{t_1}, \cdots, \widehat{t_d}\} \cap \mathcal{T}^*= \emptyset$
in the first corruption query CQ() of the adversary in Query 1 phase. 

\subsection{Background on Lattices} \label{background}
In this work, all vectors are written as columns. The transpose of a vector $\textbf{b}$ (resp., a matrix $\textbf{A}$) is denoted as $\textbf{b}^T$ (resp., $\textbf{A}^T$). The Gram-Schmidt (GS) orthogonaliation of  $\textbf{S}:=[\s_1,\cdots, \s_k]$ is denoted by $\widetilde{\textbf{S}}:=[\widetilde{\s}_1,\cdots,\widetilde{\s}_k ]$  in the same order.
 
\noindent \textbf{Lattices.} A lattice is a set  $\calL=\calL(\textbf{B}):=\left\{\sum_{i=1}^m\b_ix_i : x_i\in\ZZ~\forall i\in[m ] \right\}\subseteq\ZZ^m$ 
generated by a basis $\textbf{B}=[\textbf{b}_1|\cdots |\textbf{b}_m]\in \ZZ^{n\times m}.$ 
We are interested in the following  lattices:\\
$	\Lambda^{\bot}_q(\textbf{A}) :=\left\{ \e\in\ZZ^m \text{ s.t. } \textbf{A}\e=0 ~(\text{mod } q) \right\}, $
	$\L_q^{\textbf{u}}(\textbf{A}) :=  \left\{ \e\in\ZZ^m~\rm{s.t.}~ \textbf{A}\e=\textbf{u}~(\text{mod } q)  \right\},$\\
	$\L_q^{\textbf{U}}(\textbf{A}) :=  \left\{ \mathbf{R}\in\ZZ^{m\times k}~\rm{s.t.}~ \textbf{A}\mathbf{R}=\textbf{U} (\text{mod } q)  \right\},$
where $\textbf{A}\xleftarrow{\$}\ZZ^{n\times m}$, $\textbf{u}\in\ZZ_q^n$ and $\textbf{U}\in\ZZ_q^{n \times k}$.  

For a vector $\textbf{s}=(s_1,\cdots, s_n)$, $\|\textbf{s}\|:=\sqrt{s_1^2+\cdots+s_n^2}$, $\|\textbf{s}\|_\infty:=\max_{i \in [n]}|s_i|$.  For a matrix $\textbf{S}=[\s_1\cdots\s_k]$ and any vector $\textbf{x}=(x_1,\cdots, x_k)$,   we define $\|\textbf{S}\|:=\max_{i\in [k]}\|\s_i\|$, the GS norm of $\textbf{S}$ is $\|\widetilde{\textbf{S}}\|$, the sup norm is $\| \mathbf{S}\|_{sup}=\sup_{\textbf{x}}\frac{\| \textbf{S} \textbf{x}\|}{\| \textbf{x}\|}$. This yields for all $\textbf{x}$ that $\| \textbf{S} \textbf{x}\| \leq \| \textbf{S}\|_{sup} \cdot \|\textbf{x}\|$.
We call a basis $\textbf{S}$  of some lattice \textit{short} if $\|\widetilde{\textbf{S}}\|$ is short. \\

\noindent \textbf{Gaussian Distributions.}  
Assume $m\geq 1$, $\mathbf{v}\in \mathbb{R}^m$, $\sigma>0$, and  $\mathbf{x}\in \mathbb{R}^m$.
We define  the function $\rho_{\sigma,\mathbf{v}}(\mathbf{x})= \exp({{-\pi \Vert \mathbf{x}-\mathbf{v}\Vert^2 }/{ \sigma^2}})$. 
\begin{definition}[Discrete Gaussians]\label{def12}
	Suppose that $\mathcal{L}\subseteq\ZZ^m$ is  a lattice, and  $\mathbf{v}\in\RR^m$ and $\sigma>0$. 
	The discrete Gaussian distribution over $\mathcal{L}$ with center $\mathbf{v}$ and parameter $\sigma$ is defined by
	$\cal{D}_{\mathcal{L},\sigma,\mathbf{v}}(\mathbf{x})=\frac{\rho_{\sigma,\mathbf{v}}(\mathbf{x})}{\rho_{\sigma,\mathbf{v}}(\mathcal{L})}$ 
	for $\mathbf{x}\in\mathcal{L},$ where  $ 	\rho_{\sigma,\mathbf{v}}(\mathcal{L}):=\sum_{\x\in\mathcal{L}}\rho_{\sigma,\mathbf{v}}(\x).$
\end{definition}

\begin{lemma}[{\cite[Lemma 4.4]{MR04}}]\label{thm:Gauss}
	Let $q> 2$ and let $\mathbf{A},\mathbf{B}$ be a matrix in $\ZZ_q^{n\times m}$ with $m>n$. Let $\mathbf{T}_\mathbf{A}$  
	be a basis for  $\Lp_q(\mathbf{A})$. Then, for $\sigma\geq\|\widetilde{\mathbf{T}_\mathbf{A}}\|\cdot  \omega(\sqrt{\log n})$, 
	$\Pr[\x\gets \mathcal{D}_{\Lambda_q^{\bot}(\mathbf{A}),\sigma}:~\|\x\|>\sigma\sqrt{m}]\leq\mathsf{negl}(n).$
\end{lemma}

\subsubsection{Learning with Errors.} 
The security for our construction relies on the decision variant of the learning with errors  (DLWE) problem defined below.
\begin{definition}[DLWE, {\cite{Reg05}}] \label{lwe}
	Suppose that $n$ be a positive integer, $q$ is prime, and $\chi$ is a distribution  over $\ZZ_q$. 
	The $(n, m, q, \chi)$-$\mathsf{DLWE}$ problem requires to distinguish 
	$(\mathbf{A}, \mathbf{A}^T\mathbf{s}+\mathbf{e})$ from $(\mathbf{A}, \mathbf{c}),$
	where  $\mathbf{A} \xleftarrow{\$}\ZZ_q^{n \times m} , \s \xleftarrow{\$} \ZZ_q^n, \mathbf{e}\gets \chi^m, \mathbf{c}\xleftarrow{\$} \ZZ_q^m.$ 
\end{definition}
Let $\chi$ be a $\chi_0$-bounded noise distribution, i.e.,  its support belongs to $[-\chi_0,\chi_0]$. 
The hardness of DLWE is measured by $q/\chi_0$, which is always greater than 1 as $\chi_0$ is chosen such that $\chi_0<q$. 
Specifically, the smaller $q/\chi_0$ is, the harder DLWE is. (See \cite[Subsection 2.2]{BGG+14} and \cite[Section 3]{BV16}  for further discussions.) 

\begin{lemma}[{\cite[Corollary 3.2]{BV16}}] \label{dlwehard}
For all $\epsilon>0$, there exist functions $q=q(n)\leq 2^n$, $m=\Theta(n\log q)=\mathsf{poly}(n)$, $\chi=\chi(n)$ such that $\chi$ is a $\chi_0$-bounded for some $\chi_0=\chi_0(n)$, $q/\chi_0 \geq 2^{n^\epsilon}$ and such that $DLWE_{n, m, q, \chi}$ is at least as hard as the classical hardness of GapSVP$_{\gamma}$ and the quantum hardness of SIVP$_\gamma$ for $\gamma=2^{\Omega(n^\epsilon)}$.
\end{lemma}
The GapSVP$_{\gamma}$ problem is the one, given a basis for a lattice and a positive number $d$, requires to distinguish between two cases; (i)the lattice has a vector shorter than $d$, and (ii) all lattice vector have length bigger than $\gamma\cdot d$. And SIVP$_\gamma$ is the problem that, given a basis for a lattice of rannk $n$, requires to find a set of $n$ ``short" and independent lattice vectors.\\

\noindent \textbf{Leftover Hash Lemma.} 
The following variant of the so-called leftover hash  lemma will be used in this work to support our arguments.
\begin{lemma}[{\cite[Lemma 13]{ABB10}}] \label{lhl}
	Let $m,n,q$ be such that $m>(n+1)\log_2 q+ \omega(\log n)$ and that $q> 2$ is prime. 
	Let $\mathbf{ A}$ and $\mathbf{B }$ are uniformly chosen from   
	$\mathbb{Z}_q^{n \times m}$ and  $\mathbb{Z}_q^{n \times k}$, respectively. 
	Then for any uniformly chosen matrix $\mathbf{S}$ from $\{-1, 1\}^{m \times k} \text{ (mod } q)$
	and for all vectors $\mathbf{e}\in \mathbb{Z}_q^{m}$, 
	$$(\mathbf{A}, \mathbf{A}\mathbf{S}, \mathbf{S}^T\mathbf{e}) \stackrel{\text{s}}{\approx} (\mathbf{A}, \mathbf{B}, \mathbf{S}^T\mathbf{e}).$$
\end{lemma}

We conclude this section with some standard results regarding trapdoor 
mechanism often used in lattice-based cryptography. \\

\noindent \textbf{Lattice Trapdoor Mechanism.}	In our context, 
a (lattice) trapdoor is a short basis $\textbf{T}_\textbf{A}$ for the  $q$-ary  lattice $\Lambda^{\bot}_q(\textbf{A})$,
i.e., $\textbf{A}\cdot\textbf{T}_\textbf{A}=0 \text{ (mod } q)$ (see \cite{GPV08}). 
We call  $\textbf{T}_\textbf{A}$ the associated trapdoor for $\Lambda^{\bot}_q(\textbf{A})$ or even for $\textbf{A}$. 
\begin{lemma}\label{trapdoor}
	Let $n, m, q>0$ and $q$ be prime. 
	\begin{enumerate}
		\item $(\mathbf{A},\mathbf{T}_\mathbf{A}) \leftarrow \mathsf{TrapGen}(n,m,q)$ (\cite{AP09}, \cite{MP12}): 
		This is a PPT algorithm that  outputs a pair 
		$(\mathbf{A},\mathbf{T}_\mathbf{A}) \in \mathbb{Z}_q^{n \times m}\times \mathbb{Z}_q^{m \times m}$,
		where $\mathbf{T}_\mathbf{A}$ is a trapdoor for $\Lambda^{\bot}_q(\mathbf{A})$
		such that $\mathbf{A}$ is negligibly close to uniform and $\| \widetilde{\mathbf{T}_\mathbf{A}} \|=O(\sqrt{n \log q})$.
		The algorithm works if $m=\Theta(n \log q)$.  
		\item $\mathbf{T}_\mathbf{D}\leftarrow \mathsf{ExtBasisRight}(\mathbf{D}:=[\mathbf{A}|\mathbf{A}\mathbf{S}+\mathbf{B}], \mathbf{T}_\mathbf{B})$ 
		( \cite{ABB10}): This is a DPT algorithm that, 
		for the input $(\mathbf{D}, \mathbf{T}_\mathbf{B})$,
		outputs a trapdoor 
		$\mathbf{T}_\mathbf{D}$ for $\Lambda^{\bot}_q(\mathbf{D})$ 
		such that $\| \widetilde{\mathbf{T}_\mathbf{D}}\| \leq \| \widetilde{\mathbf{T}_\mathbf{B}}\|(1+\|\mathbf{S}\|_{sup})$,
		where $\mathbf{A}, \mathbf{B}  \in \mathbb{Z}_q^{n \times m}$.	
		\item $\mathbf{T}_\mathbf{E}\leftarrow \mathsf{ExtBasisLeft}(\mathbf{E}:=[\mathbf{A}|\mathbf{B}], \mathbf{T}_\mathbf{A})$ (\cite{CHKP10}): 
		This is a DPT algorithm that 
		for $\mathbf{E}$ of the form  $\mathbf{E}:=[\mathbf{A}|\mathbf{B}]$ and a trapdoor $\mathbf{T}_\mathbf{A}$ for $\Lambda^{\bot}_q(\mathbf{A})$, 
		outputs a trapdoor $\mathbf{T}_\mathbf{E}$ for $\Lambda^{\bot}_q(\mathbf{E})$ 
		such that $\| \widetilde{\mathbf{T}_\mathbf{E}}\| =\| \widetilde{\mathbf{T}_\mathbf{A}}\|$, where $\mathbf{A}, \mathbf{B}  \in \mathbb{Z}_q^{n \times m}$. 
		\item $\mathbf{R}\leftarrow\mathsf{SampleD}(\mathbf{A},\mathbf{T}_\mathbf{A}, \mathbf{U}, \sigma)$ ( \cite{GPV08}): 
		This is a PPT algorithm that takes 
		a matrix $\mathbf{A} \in \mathbb{Z}_q^{n \times m}$,  
		its associated  trapdoor $\mathbf{T}_\mathbf{A} \in \mathbb{Z}^{m \times m}$, 
		a matrix $\mathbf{U} \in \mathbb{Z}_q^{n\times k}$ and a real number $\sigma>0$ and returns a
		short matrix $\mathbf{R} \in \mathbb{Z}_q^{m \times k} $ chosen randomly according to a distribution 
		that is statistically close to $\mathcal{D}_{\Lambda^{\mathbf{U}}_q(\mathbf{A}),\sigma}$.
		The algorithm works if $\sigma=\| \widetilde{\mathbf{T}_\mathbf{A}}\|\cdot \omega(\sqrt{\log m})$. Furthermore, $\|\textbf{R}^T\|_{sup} \leq \sigma\sqrt{mk}$, $\|\textbf{R}\|_{sup} \leq \sigma\sqrt{mk}$ (see also in \cite[Lemma 2.5]{BGG+14}).
		 \item $\mathbf{T}'_\mathbf{A} \leftarrow\mathsf{RandBasis}(\mathbf{A},\mathbf{T}_\mathbf{A}, \sigma)$ ( \cite{CHKP10}): 
		This is a PPT algorithm that takes 
		a matrix $\mathbf{A} \in \mathbb{Z}_q^{n \times m}$,  
		its associated  trapdoor $\mathbf{T}_\mathbf{A} \in \mathbb{Z}^{m \times m}$, and a real number $\sigma>0$ 
		and returns a new basis  $\mathbf{T}'_\mathbf{A}$ for $\Lambda^{\bot}_q(\mathbf{A})$ chosen randomly according to a distribution 
		that is statistically close to $(\mathcal{D}_{\Lambda^{\bot}_q(\mathbf{A}),\sigma})^m$, and  $\| \widetilde{\mathbf{T}'_\mathbf{A}}\| \leq \sigma\sqrt{m}$.
		The algorithm works if $\sigma=\| \widetilde{\mathbf{T}_\mathbf{A}}\|\cdot \omega(\sqrt{\log m})$.
	\end{enumerate}
\end{lemma}

\section{Generic PE Construction from DFKHE} \label{generic}

\subsection{Delegatable Fully Key-homomorphic Encryption}

Delegatable fully key-homomorphic encryption (DFKHE) can be viewed as a generalised notion of 
the so-called fully key-homomorphic encryption (FKHE) \cite{BGG+14} augmented with a key delegation mechanism \cite{BGG+14}. 

Informally, FKHE enables one to transform an encryption, say $ct_\mathbf{x}$, 
of a plaintext $\mu$ under a public variable $\mathbf{x}$ into the one, say $ct_f$, 
of the same $\mu$ under  some value/function pair $(y, f)$, with the restriction that one is only able to decrypt  the ciphertext $ct_f$ if $f(\mathbf{x})=y$. 
Similarly, DFHKP together with the key delegation mechanism allows one to do the same
but with  more functions, i.e., $(y, f_1, \cdots, f_k)$, and the condition for successful decryption is that $f_1(\mathbf{x})=\cdots=f_k(\mathbf{x})=y$.

\begin{definition}[DFKHE] \label{dfkhe2} Let $\lambda, d=d(\lambda) \in \mathbb{N}$ 
be two positive integers and  let  $\mathcal{T}=\mathcal{T}(\lambda)$ and $\mathcal{Y}=\mathcal{Y} (\lambda)$ be two finite sets. 
Define $\mathcal{F}=\mathcal{F}(\lambda)=\{f| f: \mathcal{T}^{d}  \rightarrow \mathcal{Y} \}$ to be a family of efficiently computable functions. 
$(\lambda, d,$ $\mathcal{T}, \mathcal{Y}, \mathcal{F})$--DFKHE  is a tuple  
consisting of algorithms as follows.
\begin{description}
\item \underline{$( \mathsf{dfkhe}.pk, \mathsf{dfkhe}.sk) \leftarrow \mathsf{DFKHE.KGen}(1^\lambda,\mathcal{F} )$: } 
This PPT algorithm takes as input a security parameter $\lambda$ and outputs
 a public key $\mathsf{dfkhe}.pk$ and a secret key $\mathsf{dfkhe}.sk$. 
\item\underline{ $\mathsf{dfkhe}.sk_{y,f} \leftarrow \mathsf{DFKHE.KHom}(\mathsf{dfkhe}.sk, (y,f) )$: }  
This PPT algorithm takes as input the secret key $\mathsf{dfkhe}.sk$ and a pair $(y,f)\in \mathcal{Y} \times \mathcal{F} $
and returns a secret homomorphic  key $sk_{y,f}$.

\item \underline{ $\mathsf{dfkhe}.sk_{y, f_1,\cdots, f_{k+1}} \leftarrow \mathsf{DFKHE.KDel}(\mathsf{dfkhe}.pk, \mathsf{dfkhe}.sk_{y, f_1,\cdots, f_{k}}, (y,f_{k+1}) )$:}   
This PPT algorithm takes as input the public key $\mathsf{dfkhe}.pk$, a function $f_{k+1}\in \mathcal{F}$
and the secret key $\mathsf{dfkhe}.sk_{y, f_1,\cdots, f_{k}}$
and returns the delegated secret key $\mathsf{dfkhe}.sk_{y, f_1,\cdots, f_{k+1}}$. 
Further, the key $\mathsf{dfkhe}.sk_{y, f_1,\cdots, f_{k}}$ is produced either by $\mathsf{DFKHE.KHom}$ if $k=1$,
 or iteratively  by  $ \mathsf{DFKHE.KDel}$ if $k>1$.

\item \underline{$(\mathsf{dfkhe}.ct,\mathbf{t} ) \leftarrow \mathsf{DFKHE.Enc}(\mathsf{dfkhe}.pk,  \mu, \mathbf{t} )$:}    
This PPT algorithm takes as input the public key $\mathsf{dfkhe}.pk$, a plaintext $\mu$
 and a variable $\mathbf{t} \in \mathcal{T}^d$ 
 and returns a ciphertext $\mathsf{dfkhe}.ct$-- an encryption of $\mu$ under  the variable  $\mathbf{t} $.


\item \underline{$\mathsf{dfkhe}.ct_{f_1,\cdots, f_k} \leftarrow \mathsf{DFKHE.ExtEval}(f_1,\cdots, f_{k},(\mathsf{dfkhe}.ct,\mathbf{t} ))$: } 
The DPT algorithm takes as input a ciphertext $\mathsf{dfkhe}.ct$ and the associated variable $\mathbf{t} \in \mathcal{T}^d$
and returns an evaluated ciphertext $\mathsf{dfkhe}.ct_{f_1,\cdots, f_k}$. If $f_1(\mathbf{t})=\cdots=f_k(\mathbf{t})=y$ ,  then we say that $\mathsf{dfkhe}.ct_{f_1, \cdots, f_k}$ is an encryption  of $\mu$ using the public key $(y,f_1, \cdots, f_k)$.

\item \underline{$\mu/\bot \leftarrow \mathsf{DFKHE.Dec}(\mathsf{dfkhe}.sk_{y,f_1, \cdots, f_k}, (\mathsf{dfkhe}.ct,\mathbf{t}))$:}   
The DPT algorithm takes as input  a delegated secret key $\mathsf{dfkhe}.sk_{y,f_1, \cdots, f_k}$
and a ciphertext $\mathsf{dfkhe}.ct$  associated with $ \mathbf{t}\in \mathcal{T}^d$
and recovers a plaintext $\mu$. 
It succeeds if $f_i(\mathbf{t})=y$ for all $i\in[k]$. 
Otherwise, it fails and returns $\bot$. 
To recover $\mu$, the algorithm first calls $ \mathsf{DFKHE.ExtEval}(f_1,\cdots, f_{k},(\mathsf{dfkhe}.ct,\mathbf{t} ))$ 
and gets $\mathsf{dfkhe}.ct_{f_1,\cdots, f_k}$.
Next it uses $\mathsf{dfkhe}.sk_{y, f_1, \cdots, f_k}$ and opens $\mathsf{dfkhe}.ct_{f_1,\cdots, f_k}$.

\end{description}
\end{definition}
Obviously, DFKHE from Definition \ref{dfkhe2} is identical to FKHE \cite{BGG+14}  if $k=1$.\\

\noindent \textbf{Correctness.}
For all $\mu \in \mathcal{M}$, all $k\in \mathbb{N}$, all $f_1, \cdots, f_k\in \mathcal{F}$ and $\mathbf{t} \in \mathcal{T}^d$, $y\in \mathcal{Y}$, 
over the randomness of $(\mathsf{dfkhe}.pk, \mathsf{dfkhe}.sk) \leftarrow \mathsf{FKHE.KGen}(1^\lambda,\mathcal{F} )$, 
$(\mathsf{dfkhe}.ct,\mathbf{t} ) \leftarrow \mathsf{FKHE.Enc}(\mathsf{dfkhe}.pk,  \mu, \mathbf{t} )$, 
$\mathsf{dfkhe}.sk_{y,f_1} \leftarrow \mathsf{FKHE.KHom}(\mathsf{dfkhe}.sk, (y,f_1) )$ and  \\
$\mathsf{dfkhe}.sk_{y, f_1,\cdots, f_{i}}$ $\leftarrow \mathsf{FKHE.KDel}(\mathsf{dfkhe}.sk_{y, f_1,\cdots, f_{i-1}}, $ $(y,  f_{i} )),$  $\mathsf{dfkhe}.ct_{f_1,\cdots, f_k} \leftarrow \mathsf{DFKHE.ExtEval}$ $(f_1,\cdots, f_{k},(\mathsf{dfkhe}.ct,\mathbf{t} ))$ for all $ i\in \{2,\cdots, k\}$, then
 \begin{itemize}
  \item $ \Pr[\mathsf{FKHE.Dec}(\mathsf{dfkhe}.sk, (\mathsf{dfkhe}.ct, \mathbf{t}))=\mu]\geq 1-negl(\lambda),$
\item if $y=f_1(\mathbf{t})=\cdots=f_k(\mathbf{t})$, then
\begin{eqnarray*}
 && \Pr[\mathsf{FKHE.Dec}(\mathsf{dfkhe}.sk, (\mathsf{dfkhe}.ct_{f_1,\cdots, f_k}, \mathbf{t}))=\mu]\geq 1-negl(\lambda),\\
&& \Pr[\mathsf{FKHE.Dec}(\mathsf{dfkhe}.sk_{y,f_1,\cdots, f_i}, (\mathsf{dfkhe}.ct, \mathbf{t}))=\mu]\geq 1-negl(\lambda), \forall i\in [k],
\end{eqnarray*}
\item  For any $ i\in [k]$, if $ y\neq f_i(\mathbf{t}),$ 
\[
 \Pr[\mathsf{FKHE.Dec}(\mathsf{dfkhe}.sk_{y,f_1,\cdots, f_j}, (\mathsf{dfkhe}.ct, \mathbf{t}))=\mu]\leq negl(\lambda), \forall j\in \{i, k\}.
 \]
\end{itemize}

\noindent \textbf{Security.}  Security of DFKHE is similar to that of FKHE from \cite{BGG+14}
with an extra evaluation that includes the key delegation mechanisms. 
\begin{definition}[Selectively-secure CPA of DFKHE]
	DFKHE is IND-sVAR-CPA if for any polynomial time adversary $\mathcal{B}$ in the game $\mathsf{IND}$-$\mathsf{sVAR}$-$\mathsf{CPA}^{\mathsf{sel},\mathcal{B}}_{\mathsf{DFKHE}}$, the adversary advantage  
	$\mathsf{Adv}_{\mathsf{DFKHE}}^{\mathsf{IND}\text{-}\mathsf{sVAR}\text{-}\mathsf{CPA}}(\mathcal{B})=|\Pr[b'=b]-\frac{1}{2}| \leq \mathsf{negl}(\lambda).$
\end{definition}

The $\mathsf{IND}$-$\mathsf{sVAR}$-$\mathsf{CPA}^{\mathsf{sel},\mathcal{B}}_{\mathsf{DFKHE}}$ 
game  is as follows.
\begin{enumerate}
	\item \textbf{Initialize.} On the security parameter $\lambda$ and $\lambda$--dependent tuple $(d, (\mathcal{T}, \mathcal{Y}, \mathcal{F}))$,  $\mathcal{B}$ releases the target variable $\widehat{\mathbf{t}}=(\widehat{t_1}, \cdots, \widehat{t_d})\in \mathcal{T}^d$. 
	\item \textbf{Setup.} 
	The challenger runs $(\textsf{dfkhe}.pk, \textsf{dfkhe}.sk) \leftarrow \textsf{DFKHE.KGen}(1^\lambda, \mathcal{F} )$. Then, it gives $\textsf{dfkhe}.pk$ to $\mathcal{B}$.
	\item \textbf{Query.}   $\mathcal{B}$ adaptively makes delegated key queries DKQ($y,(f_1, \cdots, f_k)$) to get the corresponding delegated secret keys. 
	Specifically, $\mathcal{B}$ is allowed to have an access to the oracle  $KG(\mathsf{dfkhe}.sk,\widehat{\mathbf{t}},y,(f_1, \cdots, f_k))$, which takes as input $\mathsf{dfkhe}.sk,$ $\widehat{\mathbf{t}},$ a list of functions $f_1, \cdots, f_k\in \mathcal{F}$ and $y\in \mathcal{Y}$ 
and returns either $\bot$ if all $f_j(\widehat{\mathbf{t}})=y$, 
or the delegated secret key $\mathsf{dfkhe}.sk_{y, f_1,\cdots, f_{k}}$ otherwise. 
The delegated secret key $\mathsf{dfkhe}.sk_{y, f_1,\cdots, f_{k}}$ is computed calling $\mathsf{dfkhe}.sk_{y, f_1}: =\mathsf{DFKHE.KHom}$ $(\mathsf{dfkhe}.sk, (y,f_1))$ and\\
   $\mathsf{dfkhe}.sk_{y, f_1,\cdots, f_{i}}\leftarrow \mathsf{DFKHE.KDel}$ $(\mathsf{dfkhe}.pk, \mathsf{dfkhe}.sk_{y, f_1,\cdots, f_{i-1}}, (y,f_{i} )),~ \forall i\in \{2, \cdots, k\}$.		

	\item \textbf{Challenge.} 
	The adversary submits two messages $\mu_0, \mu_1$ (with $\widehat{\mathbf{t}}$). The challenger in turn chooses $b \xleftarrow{\$} \{0,1\}$ and 
	returns the output $(\mathsf{dfkhe}.\widehat{ct}, \widehat{\mathbf{t}})$ of $ \textsf{DFKHE.Enc}(\mathsf{dfkhe}.pk,\mu_b, $ $\widehat{\mathbf{t}})$.  
			\item \textbf{Guess.} The adversary outputs $b'\in \{0,1\}$. It wins if $b'=b$. 
\end{enumerate}

\subsection{Generic PE Construction from DFKHE.} 
The main idea behind our construction is an observation that 
ciphertext tags  can be treated as  
 variables  $\textbf{t}=(t_1, \cdots, t_d) \in \mathcal{T}^d$. The puncturing property, which is related to the ``equality", suggests us to construct a family $\mathcal{F}$ of equality test functions, allowing to compare each pair of ciphertext tags  and punctures. 
 Using this idea, we then can have a PE  construction from DFKHE.

Let $\lambda, d=d(\lambda) \in \mathbb{N}$ be two positive integers. 
Let $\mathcal{T}=\mathcal{T}({\lambda})$ be a finite set (that henceforth called the \textit{tag space}) and $\mathcal{Y}=\mathcal{Y}({\lambda})$ 
be also a finite set. In addition, let  $y_0\in \mathcal{Y}$ be a some fixed special element. 
Define a family of all equality test functions indicated by $\mathcal{T}$,  

\begin{equation}\label{eq10}
\mathcal{F}=\mathcal{F}({\lambda}):=\left \{ f_{t^*}| t^* \in \mathcal{T}, \forall  \mathbf{t}=(t_1, \cdots, t_d), f_{t^*}: \mathcal{T}^d  \rightarrow \mathcal{Y}  \right \},
\end{equation}
where $f_{t^*}(\mathbf{t}):=y_0$  if $t^* \neq  t_i, \forall i\in[d]$, $f_{t^*}(\mathbf{t}):= y_{t^*,\mathbf{t}}\in \mathcal{Y}\setminus \{y_0\}$.  
Here, $y_{t^*,\mathbf{t}}$ means depending on the value of $t^*$ and $\mathbf{t}$. 
Now, let  $\Pi=(\mathsf{DFKHE.KGen},  \mathsf{DFKHE.KHom},  $ $ \mathsf{DFKHE.Enc}, $  $ \mathsf{DFKHE.ExtEval},$ $ \mathsf{DFKHE.KDel},$ $\mathsf{DFKHE.Dec} )$ be  $(\lambda, d,$ $\mathcal{T}, \mathcal{Y}, \mathcal{F})$--DFHKE. 
   Using  $\Pi$, we can construct a PE system $\Psi=(\mathsf{PE.key}, \mathsf{PE.enc}, \mathsf{PE.pun},$ $ \mathsf{PE.dec} )$ of which both tags and punctures reside in $ \mathcal{T}$. The description of $\Psi$ is below:
\begin{description}
	\item \underline{$(\mathsf{pe}.pk, \mathsf{pe}.sk_0) \leftarrow \mathsf{PE.key}(1^\lambda, d)$:}  For input a security parameter $\lambda$ and
	the maximum number $d$ of tags per ciphertext,   
	run $(\mathsf{dfkhe}.pk, \mathsf{dfkhe}.sk) $ $\leftarrow \mathsf{DFKHE.KGen}(1^\lambda, \mathcal{F})$, and 
	return  $\mathsf{pe}.pk:=\mathsf{dfkhe}.pk$, and $ \mathsf{pe}.sk_0:=\mathsf{dfkhe}.sk$.
	\item \underline{$\mathsf{pe}.ct \leftarrow \mathsf{PE.enc}(\mathsf{pe}.pk,\mu, \mathbf{t}=(t_1, \cdots, t_d))$:} For a public key $\mathsf{pe}.pk$, a message $\mu$, and ciphertext tags $\mathbf{t}=(t_1, \cdots, t_d)$,  return $\mathsf{pe}.ct \leftarrow \mathsf{DFKHE.Enc}(\mathsf{pe}.pk, \mu, \mathbf{t} )$.
	
	
	\item \underline{$\mathsf{pe}.sk_{i} \leftarrow \mathsf{PE.pun}(\mathsf{pe}.pk, \mathsf{pe}.sk_{i-1}, t^*_{i})$:} 
	For input $\mathsf{pe}.pk$, $\mathsf{pe}.sk_{i-1}$ and a punctured tag $t^*_i$,
	
	\begin{itemize}
	\item  If $i=1$: run  $\mathsf{dfkhe}.sk_{y_0,f_{t^*_1}} \leftarrow \mathsf{DFKHE.KHom}(\mathsf{pe}.sk_{0}, (y_0,f_{t^*_1}))$ and output $\mathsf{pe}.sk_{1}:=\mathsf{dfkhe}.sk_{y_0,f_{t^*_1}} $.
	\item If $i  \geq 2$: compute 
	$\mathsf{pe}.sk_{i} \leftarrow \mathsf{DFKHE.KDel}(\mathsf{dfkhe}.pk, \mathsf{pe}.sk_{i-1},(y_0,f_{t^*_{i}}) ).$
	\item Finally, output $\mathsf{pe}.sk_{i}.$
	\end{itemize}
	
	\item  \underline{$\mu/\bot \leftarrow \mathsf{PE.dec}(\mathsf{pe}.pk, (\mathsf{pe}.sk_{i}, (t^*_1, \cdots, t^*_i)),(\mathsf{pe}.ct, \mathbf{t}))$:}  
	For input the public key
	$\mathsf{pe}.pk$, a puncture key $\mathsf{pe}.sk_{i}$ together with punctures  $(t^*_1, \cdots, t^*_i)$, a ciphertext $\mathsf{pe}.ct$ and its associated tags $\mathbf{t}=(t_1, \cdots, t_d)$, the algorithm first checks whether or not $f_{t^*_1}(\mathbf{t})=\cdots=f_{t^*_i}(\mathbf{t})=y_0$. If not, the algorithm  returns $\bot$. Otherwise, it returns the output of $ \mathsf{DFKHE.Dec}(\mathsf{pe}.sk_{i}, \mathsf{pe}.ct)$. 
	
\end{description}

\noindent \textbf{Correctness.}  Remark that, over the choice of  $(\lambda, d, \eta,  (t^*_1, \cdots, t^*_\eta),$ $ (t_1, \cdots, t_d)$,  $\eta\geq 0$, $t^*_1, \cdots, t^*_\eta \in \mathcal{T}$, $t_1, \cdots, t_d \in \mathcal{T}\setminus \{t^*_1, \cdots, t^*_\eta\}$, we have  $f_{t^*_j}(\mathbf{t})=y_0$ for all $j\in [\eta].$ Then, it is clear that, the induced PE $\Psi$ is correct if  and only if the DFKHE $\Pi$ is correct.

\begin{theorem} \label{pe}
PE $\Psi$ is selectively-secure CPA assuming that 
the underlying DFKHE $\Pi$ is selectively-CPA secure.
\end{theorem}
\begin{proof} \label{peproof}
	Assume that there exists an adversary $\mathcal{A}$ that is able to break the selective security of $\Psi$ 
	with probability $\delta$. We can construct a simulator $\mathcal{S}$, which takes advantage of  $\mathcal{A}$ 
	and breaks selective security of $\Pi$ with the same probability. 
	\begin{description}
		\item \textbf{Initialize.} $\mathcal{S}$ would like to break the selective security of the   ($\lambda, d, (\mathcal{T}, \mathcal{Y}, \mathcal{F}$)--DFHKE system $\Pi=( \mathsf{DFKHE.KGen},  \mathsf{DFKHE.KHom},  \mathsf{DFKHE.Enc},$ $ \mathsf{DFKHE.Dec} ,$  $ \mathsf{DFKHE.ExtEval}, $ $\mathsf{DFKHE.KDel})$, 
		where  $\lambda, d, \mathcal{T}, $ $\mathcal{Y},$ $ \mathcal{F}$ are specified as in and around Equation \eqref{eq10}.
		
		\item \textbf{Targeting.} $\mathcal{S}$ calls $\mathcal{A}$ to get  the target tags  $(\widehat{t}_1, \cdots, \widehat{t}_d)$ in the game for $\Psi$, and lets it be $\widehat{\mathbf{t}}$, playing the role of the target variable in the game for $\Pi$.
		
		\item \textbf{Setup.} $\mathcal{S}$ initializes a set of punctured tags $\mathcal{T}^* \leftarrow \emptyset$, 
	and a set of corrupted tags $\mathcal{C}^* \leftarrow \emptyset$ containing all punctured tags at the time of the first corruption query.  runs $(\mathsf{dfkhe}.pp, \mathsf{dfkhe}.pk, $ $\mathsf{dfkhe}.sk) \leftarrow \mathsf{DFKHE.KGen}(1^\lambda, \mathcal{F})$ and gives $\mathsf{dfkhe}.pp, \mathsf{dfkhe}.pk$ to $\mathcal{A}$.  Note that $ \mathsf{PE.key}(1^\lambda, d)$  $\equiv \mathsf{DFKHE.KGen}(1^\lambda, \mathcal{F})$ by construction.
		
		\item \textbf{Query 1.} In this phase, $\mathcal{A}$ adaptively makes puncture queries PQ($k, t^*_k$), 
where $k$ implicitly counts the number of PQ queries so far, and corruption queries CQ().  
To reply  PQ($k, t^*_k$), $\mathcal{S}$ simply returns the output of $\mathsf{DFKHE.KDel}(\mathsf{dfkhe}.pk, \mathsf{dfkhe}.sk_{y_0,f_{t^*_1},\cdots, f_{t^*_{k-1}}} )$, with noting that when $k=1$, then we have both $\mathsf{dfkhe}.sk_{y_0,f_{t^*_1},\cdots, f_{t^*_{k-1}}}:=\mathsf{dfkhe}.sk_{0}$ and $\mathsf{DFKHE.KDel}\equiv \mathsf{DFKHE.KHom}$ and finally appends $t^*_k$ to $\mathcal{T}^*$. 
		
		The simulator just cares about the time at which the first CQ() has been completed. At that time, $\mathcal{S}$ saves the value of the counter $k$ and makes $\mathcal{A}$'s puncture queries a list of functions $\{f_{t^*_1},\cdots, f_{t^*_{k}}\}$ and  sets $\mathcal{C}^* \leftarrow \mathcal{T}^* $. We can consider that $\mathcal{A}$ has made a sequence of $k$ queries to the $KG(\mathsf{dfkhe}.sk,\widehat{\mathbf{t}},y,(f_1, \cdots, f_k))$ oracle in the DFKHE's security game. 
Recall that, the requirement for a query to $KG$ to be accepted is that it must be \textit{not} all $j\in [k]$ sastyfying $f_j(\widehat{\mathbf{t}})=y$. This requirement is essentially fulfilled thanks to the condition in the FE's security game that there is at least one $t^*_j\in \{\widehat{t}_1, \cdots, \widehat{t}_d\}\cap \mathcal{C}^*$. 
		
		\item \textbf{Challenge.} 
		$\mathcal{A}$ submits two messages $\mu_0, \mu_1$ (with $\widehat{\mathbf{t}}$). 
		$\mathcal{S}$ in turn chooses $b \xleftarrow{\$} \{0,1\}$ and 
		returns $(\mathsf{dfkhe}.\widehat{ct},\widehat{\mathbf{t}}) \leftarrow \textsf{DFKHE.Enc}(\mathsf{dfkhe}.pk,\mu_b, \widehat{\mathbf{t}})$.  
		\item \textbf{Query 2.} The same as Query 1.
		\item \textbf{Guess.} $\mathcal{S}$ outputs the same $b'\in \{0,1\}$ as $\mathcal{A}$ has guessed. 
	\end{description}
	
	It is clear that the FE adversary $\mathcal{A}$ is joining the DFKHE game, however it is essentially impossible to distinguish the DFKHE game from the FE one  as the simulated environment for $\mathcal{A}$ is \textit{perfect}.  
This concludes the proof.\qed
\end{proof}

\section{DFKHE and FE Construction from Lattices} \label{instan}

At first, in Subsection \ref{gad} below, we will review the key-homomorphic mechanism,
which is an important ingredient for our lattice-based construction. 

\subsection{Key-homomorphic Mechanism for Arithmetic Circuits} \label{gad}
Let $n, q>0$, $k:=\lceil \log q \rceil$ and $m:=n\cdot k$. 
We exploit the gadget matrix $\textbf{G}$ and its associated trapdoor $\textbf{T}_{\textbf{G}}$.
According to \cite[Section 4]{MP12}, the matrix $\textbf{G}:=\textbf{I}_n\otimes\textbf{g}^T\in \mathbb{Z}_q^{n \times m}$, 
where $\textbf{g}^T=[1 \; 2\; 4 \; \cdots\;  2^{k-1}]$. 
The associated trapdoor $\textbf{T}_{\textbf{G}}\in \mathbb{Z}^{m \times m}$ is publicly known and 
$\| \widetilde{\textbf{T}_{\textbf{G}}}\| \leq \sqrt{5}$ (see \cite[Theorem 4.1]{MP12}). 

\noindent \textbf{Key-homomorphic Mechanism.} 
We recap some basic facts useful for construction of 
evaluation algorithms for the family of polynomial depth and unbounded fan-in arithmetic circuits
(see  \cite[Section 4]{BGG+14} for details).
Let $\textbf{G}\in \mathbb{Z}_q^{n\times m}$ be the gadget matrix given above.    
For $x\in \mathbb{Z}_q$, $\textbf{B} \in \mathbb{Z}_q^{n\times m}$, $\textbf{s} \in \mathbb{Z}_q^{n}$ and $\delta>0$, 
define the following set
$ E_{\textbf{s},\delta}(x, \textbf{B}):=\{(x\textbf{G}+\textbf{B})^T\textbf{s}+\textbf{e}, \text{ where } \| \textbf{e}\|<\delta\}.$ More details can be found in  \cite{BGG+14}.

\begin{lemma}[{\cite[Section 4]{BGG+14}}]\label{eval} 
Let $n$, $q=q(n)$, $m=\Theta(n\log q)$ be positive integers,  
	$\mathbf{x}=(x_1, \cdots, x_d) \in \mathbb{Z}_q^d$, 
	$\mathbf{x}^*=(x_1^*, \cdots, x^*_d) \in \mathbb{Z}_q^d$, $\mathbf{B}_i\in \mathbb{Z}_q^{n\times m}$, 
	$\mathbf{c}_i \in E_{\mathbf{s},\delta}(x_i, \mathbf{B}_i)$ for some $\mathbf{s}\in \mathbb{Z}_q^n$ 
	and $\delta>0$, $ \mathbf{S}_i \in \mathbb{Z}_q^{m\times m}$ for all $i\in [d]$. Also, let $\beta_{\mathcal{F}}=\beta_{\mathcal{F}}(n):\mathbb{Z} \rightarrow \mathbb{Z}$ be a positive integer-valued function, and   
		$\mathcal{F}=\{f:(\mathbb{Z}_q)^d \rightarrow \mathbb{Z}_q\}$ be a family of functions, in which each function can be computed by some circuit of a family of depth $\tau$, polynomial-size arithmetic circuits $(C_{\lambda})_{\lambda\in \mathbb{N}}$. 
	Then there exist DPT algorithms $\mathsf{Eval}_\mathsf{pk}$,   $\mathsf{Eval}_\mathsf{ct}$,  
	$ \mathsf{Eval}_\mathsf{sim}$  associated with $\beta_{\mathcal{F}}$ and $\mathcal{F}$  such that the following properties hold.
	\begin{enumerate}
		\item If $\mathbf{B}_f \leftarrow \mathsf{Eval}_\mathsf{pk}(f\in \mathcal{F}, (\mathbf{B}_i )_{i=1}^d )$, 
		then  $\mathbf{B}_f\in \mathbb{Z}_q^{n\times m}$.
		\item Let $\mathbf{c}_f \leftarrow \mathsf{Eval}_\mathsf{ct}(f\in \mathcal{F}, ((x_i, \mathbf{B}_i,\mathbf{c}_i))_{i=1}^d)$,
		then $\mathbf{c}_f \in E_{\mathbf{s},\Delta}(f(\mathbf{x}), \mathbf{B}_f)$, 
		in which $\mathbf{B}_f \leftarrow \mathsf{Eval}_\mathsf{pk}(f, (\mathbf{B}_i)_{i=1}^d)$ and $\Delta < \delta \cdot \beta_{\mathcal{F}}.$
		\item The output $\mathbf{S}_f \leftarrow \mathsf{Eval}_\mathsf{sim}(f\in \mathcal{F}, ((x_i^*,\mathbf{S}_i))_{i=1}^d, \mathbf{A})$ satisfies the relation
		$\mathbf{A}\mathbf{S}_f-f(\mathbf{x}^*)\mathbf{G}=\mathbf{B}_f$ 
		and $\|\mathbf{S}_f \|_{sup} <\beta_{\mathcal{F}}$ with overwhelming probability,  
		where  $\mathbf{B}_f \leftarrow \mathsf{Eval}_\mathsf{pk}(f, (\mathbf{A}\mathbf{S}_i-x_i^*\mathbf{G})_{i=1}^d) $. 
		In particular, if $\mathbf{S}_1,\cdots, \mathbf{S}_d \xleftarrow{\$} \{-1,1\}^{m\times m}$, 
		then $\|\mathbf{S}_f \|_{sup} <\beta_{\mathcal{F}}$ with all but negligible probability for all $f\in \mathcal{F}$.
	\end{enumerate}
\end{lemma}
 
In general, for a family $\mathcal{F}$ of functions 
represented by polynomial-size and unbounded fan-in circuits of depth $\tau$, the function $\beta_{\mathcal{F}}$ is given by the following lemma.
\begin{lemma}[{\cite[Lemma 5.3]{BGG+14}}]\label{eval2} Let $n$, $q=q(n)$, 
$m=\Theta(n\log q)$ be positive integers. 
Let $\mathcal{C}_{\lambda}$ be a family of polynomial-size arithmetic circuits of depth $\tau$ 
and $\mathcal{F}=\{f:(\mathbb{Z}_q)^d \rightarrow \mathbb{Z}_q\}$ be the set of functions $f$ that can be computed by some circuit $\mathcal{C} \in \mathcal{C}_{\lambda}$ as stated in Lemma \ref{eval}. Also, suppose that all (but possibly one)  of the input values to the multiplication gates are bounded by $p<q$. Then, 	  $\beta_{\mathcal{F}}=(\frac{p^d-1}{p-1}\cdot m)^\tau \cdot 20\sqrt{m}=O((p^{d-1}m)^\tau\sqrt{m})$.
\end{lemma}
 
\begin{definition}[FKHE enabling functions] The tuple
$(\mathsf{Eval}_\mathsf{pk}$,   $\mathsf{Eval}_\mathsf{ct}$,  
	$ \mathsf{Eval}_\mathsf{sim})$ together with the family $\mathcal{F}$ and the function $\beta_{\mathcal{F}}=\beta_{\mathcal{F}}(n)$ in the Lemma \ref{eval}  is called $\beta_{\mathcal{F}}$-FKHE enabling for the family $\mathcal{F}$.
\end{definition}

\subsection{LWE-based DFKHE Construction} \label{dfkhe}
Our LWE-based DFKHE construction $\Pi$ is adapted from LWE--based FKHE and the key delegation mechanism, 
both of which proposed in \cite{BGG+14}. Roughly speaking,  the key delegation mechanism in the lattice setting is triggered 
using the algorithms $\mathsf{ExtBasisLeft}$ and $\mathsf{ExtBasisRight}$ and $\mathsf{RandBasis}$ in Lemma \ref{trapdoor}.
Formally, LWE-based DFKHE $\Pi$ consists of the following algorithms:
\begin{description}	
\item \underline{\textbf{Parameters}:} Let $\lambda \in \mathbb{N}$ be a security parameter.  
Set $n=n(\lambda)$, $q=q(\lambda)$ and  $d=d(\lambda)$ to be fixed such that $d<q$. 
Let $\eta \in \mathbb{N}$ be the maximum number of variables that can be delegated and $\sigma_1, \cdots, \sigma_{\eta}$ 
be Gaussian parameters. Also, we choose a constant $\epsilon \in (0,1)$,
 which is mentioned in Lemma \ref{dlwehard}. 
 The constant is used to determine the tradeoff between the security level and the efficiency of the system. Let $\mathcal{F}:=\{ f| f: ( \mathbb{Z}_q)^d \rightarrow \mathbb{Z}_q\}$ be a family of efficiently computable functions over $\mathbb{Z}_q$ that can be computed by some circuit of a family of depth $\tau$, polynomial-size arithmetic circuits $(C_{\lambda})_{\lambda\in \mathbb{N}}$. Take the algorithms  $(\mathsf{Eval}_\mathsf{pk}$,   $\mathsf{Eval}_\mathsf{ct}$,  $ \mathsf{Eval}_\mathsf{sim})$ together with a function $\beta_{\mathcal{F}}=\beta_{\mathcal{F}}(n)$ to be $\beta_{\mathcal{F}}$--FKHE enabling for $\mathcal{F}$.
\item \underline{$\textsf{DFKHE.KGen}(1^\lambda, \mathcal{F} )$:} 
	For the input pair (a security parameter $\lambda \in \mathbb{N}$ and a family $\mathcal{F}$) \footnote{Here, $d$ also appears implicitly as an input.}, do the following:
	\begin{enumerate}
		\item Choose  $m=\Theta{(n\log q)}$. The plaintext space is $\mathcal{M}:= \{0,1\}^m$,  $\mathcal{T}:=\mathbb{Z}_q$. Additionally, let $\chi$ be a $\chi_0$--bounded noise distribution (i.e, its support belongs to $[-\chi_0, \chi_0]$) for which the $(n,2m, q,\chi)$--DLWE is hard. 
		\item Generate $(\textbf{A},\textbf{T}_\textbf{A}) \leftarrow \textsf{TrapGen}(n,m,q)$,  sample $\textbf{U}, \textbf{B}_1, \cdots ,\textbf{B}_d \xleftarrow{\$} \mathbb{Z}_q^{n \times m}$. 
		\item Output  the public key $pk=\{\textbf{A},\textbf{B}_1, \cdots \textbf{B}_d , \textbf{U}\}$ and the initial secret key $sk=\{\textbf{T}_\textbf{A}\}$.
	\end{enumerate}
	

	\item \underline{$\textsf{DFKHE.KHom}(sk, (y,f_1) )$:} 
	For the input pair (the initial secret key  $sk$ and a pair $(y, f_1) \in \mathbb{Z}_q \times \mathcal{F}$) do the following:
	\begin{enumerate}
	\item $\textbf{B}_{f_1} \leftarrow \textsf{Eval}_\textsf{pk}(f_1, (\textbf{B}_k)_{k=1}^d)$,  $\textbf{E}_{y,f_1} \leftarrow \textsf{ExtBasisLeft}([\textbf{A}|y\mathbf{G}+\textbf{B}_{f_1}], \textbf{T}_\textbf{A})$. 
	\item   $\textbf{T}_{y,f_1} \leftarrow \textsf{RandBasis}([\textbf{A}|y\mathbf{G}+\textbf{B}_{f_1}],\textbf{E}_{y,f_1}, \sigma_1)$, output the secret key  $sk_{y,f_1}=\{\textbf{T}_{y,f_1}\}$.
Here, we set $\sigma_1=\omega(\beta_\mathcal{F}\cdot \sqrt{\log (2m)})$  for the security proof to work. 
	\end{enumerate}
	
		
		\item \underline{$\textsf{DFKHE.KDel}(sk_{y,f_1,\cdots, f_{\eta-1}}, (y, f_{\eta}) )$: }
	For the input pair (the delegated secret key  $sk_{y,f_1,\cdots, f_{\eta-1}} $ and a pair $(y, f_\eta) \in \mathbb{Z}_q \times \mathcal{F}$) do the following:
	\begin{enumerate}
	\item $\textbf{B}_{f_\eta} \leftarrow \textsf{Eval}_\textsf{pk}(f_\eta, (\textbf{B}_k)_{k=1}^d)$.
	\item   $\textbf{E}_{y,f_1,\cdots, f_{\eta}} \leftarrow \textsf{ExtBasisLeft}([\textbf{A}|y\mathbf{G}+\textbf{B}_{f_1}|\cdots|y\mathbf{G}+\textbf{B}_{f_{\eta-1}}|y\mathbf{G}+\textbf{B}_{f_\eta}], \textbf{T}_{y,f_1,\cdots, f_{\eta-1}})$.
	\item  $\textbf{T}_{y,f_1,\cdots, f_{\eta}} \leftarrow \textsf{RandBasis}([\textbf{A}|y\mathbf{G}+\textbf{B}_{f_1}|\cdots|y\mathbf{G}+\textbf{B}_{f_{\eta-1}}|y\mathbf{G}+\textbf{B}_{f_\eta}], \textbf{E}_{y,f_1,\cdots, f_{\eta}}, \sigma_\eta)$.
	\item Output the secret key $sk_{y,f_1,\cdots, f_{\eta}}=\{\textbf{T}_{y,f_1,\cdots, f_{\eta}}\}$.\\
	 We set $\sigma_\eta=\sigma_1\cdot(\sqrt{m\log m})^{\eta-1}$ and discuss on setting parameters in details later.
	\end{enumerate}

	
\item \underline{$\textsf{DFKHE.Enc}(\mu,  pk, \mathbf{t})$: }
	For the input consiting of (a message $\mu=(\mu_1, \cdots, \mu_m)\in \mathcal{M}$,  the public key $pk$ and ciphertext tags $\mathbf{t}=(t_1, \cdots, t_d) \in \mathcal{T}^d$), 
	perform the following steps:
	\begin{enumerate}
		\item Sample $\textbf{s} \xleftarrow{\$}\mathbb{Z}_q^{n}$, $\mathbf{e}_{\textsf{out}} , \textbf{e}_\textsf{in} \leftarrow \chi^m$, and $\textbf{S}_1, \cdots, \textbf{S}_{d} \xleftarrow{\$} \{-1,1\}^{m \times m}$.
		\item Compute $\textbf{e}\leftarrow (\textbf{I}_m|\textbf{S}_1|\cdots|\textbf{S}_{d})^T\textbf{e}_\textsf{in}=(\textbf{e}_{ \textsf{in}}^T, \textbf{e}_{1}^T, \cdots, \textbf{e}_{d}^T)^T $.
		\item Form $\textbf{H}\leftarrow [\textbf{A}|t_1\textbf{G}+\textbf{B}_1|\cdots |t_d \textbf{G}+\textbf{B}_d]$ 
		and compute $\textbf{c}=\textbf{H}^T\textbf{s}+\textbf{e} \in \mathbb{Z}_q^{(d+1)m}$ ,\\
		$\textbf{c}=[\textbf{c}_\textsf{in}|\textbf{c}_1|\cdots |\textbf{c}_d]$, where $\textbf{c}_{\textsf{in}}=\textbf{A}^T \textbf{s}+\textbf{e}_{\textsf{in}}$
		and  $\textbf{c}_{i}=(t_i\textbf{G}+\textbf{B}_i)^T \textbf{s}+\textbf{e}_{i}$ for $i\in [d]$.
		\item Compute $\textbf{c}_{\textsf{out}} \leftarrow \textbf{U}^T \textbf{s}+\textbf{e}_{\textsf{out}}+\mu \lceil \frac{q}{2} \rceil$.
		\item Output the ciphertext $(ct_\mathbf{t}= (\textbf{c}_{\textsf{in}}, \textbf{c}_1, \cdots, \textbf{c}_d, \textbf{c}_{\textsf{out}}), \mathbf{t}) $.
	\end{enumerate}

	\item \underline{$\textsf{DFKHE.ExtEval}(f_1,\cdots, f_\eta, ct_\mathbf{t})$: }
	For the input (a ciphertext $ct_\mathbf{t}=(\textbf{c}_{\textsf{in}}, \textbf{c}_1, $ $ \cdots, \textbf{c}_{d}, \textbf{c}_{\textsf{out}}) $ and its associated tags $\textbf{t}=(t_1, \cdots, t_d)$, 
	and a list of functions $f_1,\cdots, f_\eta \in \mathcal{F}$), execute the following steps:
	\begin{enumerate}
	
		\item Evaluate $\textbf{c}_{f_j}\leftarrow \textsf{Eval}_\textsf{ct}(f_j,  ((t_k, \textbf{B}_k, \textbf{c}_{k}))_{k=1}^{d})$ for $j\in[\eta]$.
		\item Output the evaluated ciphertext $\textbf{c}_{f_1,\cdots, f_\eta}:=(\textbf{c}_{f_1},\cdots, \textbf{c}_{f_\eta})$.
	
	\end{enumerate}

	
	\item \underline{$\textsf{DFKHE.Dec}(ct_\mathbf{t}, sk_{y,f_1,\cdots, f_\eta} )$:} 
	For the input (a ciphertext $ct_\mathbf{t}=(\textbf{c}_{\textsf{in}}, \textbf{c}_1, $ $ \cdots, \textbf{c}_{d}, \textbf{c}_{\textsf{out}}) $, 
	the associated tags $\textbf{t}=(t_1, \cdots, t_d)$, 
	and a delegated secret key $sk_{y,f_1,\cdots, f_\eta}$, execute the following steps:
	\begin{enumerate}
		\item If  $\exists j\in [\eta]$ s.t. $f_{j}(\mathbf{t}) \neq y$, then output $\bot$. Otherwise, go to Step 2.
		\item Sample $\textbf{R} \leftarrow \textsf{SampleD}([\textbf{A} |y\mathbf{G}+\textbf{B}_{f_{1}}|\cdots |y\mathbf{G}+\textbf{B}_{f_{\eta}}], \textbf{T}_{y,f_1,\cdots, f_{\eta}}, \textbf{U}, \sigma_\eta)$.
		\item Evaluate  $(\textbf{c}_{f_1},\cdots, \textbf{c}_{f_\eta}) \leftarrow \textsf{DFKHE.ExtEval}(f_1,\cdots, f_\eta, ct_\mathbf{t})$.
		\item Compute $\bar{\mu}:=(\bar{\mu}_1,\cdots, \bar{\mu}_m) \leftarrow \textbf{c}_{\textsf{out}}-\textbf{R}^T(\textbf{c}_\textsf{in}|\textbf{c}_{f_{1}}|\cdots|\textbf{c}_{f_\eta})$.
		\item For $\ell \in [m]$,  if $ |\bar{\mu}_\ell | <q/4$ then output $\mu_\ell=0$;  otherwise, output $\mu_\ell=1$.
	\end{enumerate}  
\end{description}


In the following, we will demonstrate the correctness and the security of the LWE-based DFKHE $\Pi$.

\begin{theorem}[Correctness of $\Pi$] \label{theo2}
	The proposed DFKHE $\Pi$ is correct 
	if the condition \begin{equation}\label{key}
	(\eta+1)^2\cdot \sqrt{m}\cdot \omega( (\sqrt{m\log m})^{\eta})\cdot \beta_{\mathcal{F}}^2+2<\frac{1}{4}(q/\chi_0)
	\end{equation} holds, assumming  that $f_j(\mathbf{t})=y$ for all $j\in [\eta]$. 
\end{theorem}

\begin{proof}
We have
	$ \bar{\mu}= \textbf{c}_{\textsf{out}}-\textbf{R}^T(\textbf{c}_\textsf{in}|\textbf{c}_{f_{1}}|\cdots|\textbf{c}_{f_{\eta}})=\mu \lceil \frac{q}{2} \rceil+ \textbf{e}_{\textsf{out}} -\textbf{R}^T(\textbf{e}_{\textsf{in}}|\textbf{e}_{f_1}|\cdots|\textbf{e}_{f_\eta}).$
Next, we evaluate the norm of $e_{\textsf{out}} -\textbf{R}^T(\textbf{e}_{\textsf{in}}|\textbf{e}_{f_1}|\cdots|\textbf{e}_{f_\eta})$. Since $\textbf{c}_{f_j} \in E_{\textbf{s}, \Delta}(y, \textbf{B}_{f_\eta})$, for all $j \in [\eta]$, where $\Delta<\chi_0\cdot \beta_{\mathcal{F}}$, then $\|(\textbf{e}_{\textsf{in}}|\textbf{e}_{f_1}|\cdots|\textbf{e}_{f_\eta})\| \leq \eta\cdot \Delta+\chi_0\leq (\eta\cdot \beta_{\mathcal{F}}+1)\chi_0.$ 	
	Then
	\begin{equation*}
	\begin{split}
	\|\mathbf{e}_{\textsf{out}} -\textbf{R}^T(\textbf{e}_{\textsf{in}}|\textbf{e}_{f_1}|\cdots|\textbf{e}_{f_\eta})\|_\infty &\leq \|\mathbf{e}_{\textsf{out}}\|_\infty+ \| \textbf{R}^T\|_{sup} \cdot\|(\textbf{e}_{\textsf{in}}|\textbf{e}_{f_1}|\cdots|\textbf{e}_{f_\eta})\|\\
	&\leq ((\eta+1)^2\cdot \sqrt{m}\cdot \omega( (\sqrt{m\log m})^{\eta})\cdot \beta_{\mathcal{F}}^2+2)\cdot \chi_0,\\
\end{split}
	\end{equation*}
where $\| \textbf{R}^T\|_{sup} \leq (\eta+1)m\sigma_\eta$ by Item 4 of Lemma \ref{trapdoor} and $\sigma_\eta=\sigma_1\cdot(\sqrt{m\log m})^{\eta-1}=\omega(\beta_{\mathcal{F}}\cdot \sqrt{\log m})\cdot(\sqrt{m\log m})^{\eta-1}$.
 
By choosing parameters such that $((\eta+1)^2\cdot \sqrt{m}\cdot \omega( (\sqrt{m\log m})^{\eta})\cdot \beta_{\mathcal{F}}^2+2)\cdot \chi_0<q/4$, which yields Equation \eqref{key},
	then the decryption is successful.
 \qed
\end{proof}

\begin{theorem}[IND-sVAR-CPA of $\Pi$]  \label{varcpa}
	Assuming the hardness of $(n,2m,q,\chi)$--$\mathsf{DLWE}$, the proposed DFKHE $\Pi$ is IND-sVAR-CPA.
\end{theorem}
\begin{proof} 
	The proof consists of a sequence of four games, in which the first Game 0 is 
	the original  $\mathsf{IND}$-$\mathsf{sVAR}$-$\mathsf{CPA}^{\mathsf{sel},\mathcal{A}}_{\Psi}$ game.
	The last game chooses the challenge ciphertext uniformly at random.
	Hence, the advantage of the adversary in the last game is zero.
	The games 2 and 3 are indistinguishable thanks to a reduction from the DLWE hardness. 
	\begin{description}
		\item  \textbf{Game 0.} This is the original  $\mathsf{IND}$-$\mathsf{sVAR}$-$\mathsf{CPA}^{\mathsf{sel},\mathcal{A}}_{\Psi}$ 
		game being played by an adversary $\mathcal{A}$ and a challenger. 
		At the initial phase, $\mathcal{A}$ announces a target variable  $\widehat{\mathbf{t}}=(\widehat{t_1}, \cdots, \widehat{t_d})$. 
		Note that, the challenger has to reply delegate key queries DKQ$(y, f_1,\cdots, f_k)$. However, if $(y,(f_1,\cdots, f_k)) \in \mathbb{Z}_q\times \mathcal{F}^k$ such that $f_1(\widehat{\mathbf{t}})=\cdots=f_1(\widehat{\mathbf{t}})=y$ then the query will be aborted.
		
		At the setup phase, the challenger generates $pk=\{\textbf{A},\textbf{B}_1, \cdots \textbf{B}_d , \textbf{U}\}$, 
		the initial secret key $sk=\{\textbf{T}_\textbf{A}\}$, where $\textbf{B}_1, \cdots \textbf{B}_d \xleftarrow{\$} \mathbb{Z}_q^{n \times m}$, 
		$\textbf{U} \xleftarrow{\$} \mathbb{Z}_q^{n\times m}$, $(\textbf{A}, \textbf{T}_\textbf{A})$ $\leftarrow$ $\textsf{TrapGen}(n,m,q)$. 
		The challenger then sends $pk$ to the adversary, while it keeps $sk$ secret. 
		Also, in order to produce the challenge ciphertext $\widehat{ct}$ in the challenge phase, 
		$\widehat{\textbf{S}}_1, \cdots, \widehat{\textbf{S}}_{d} \in \{-1,1\}^{m \times m}$ are generated (Step 2 of \textsf{DFKHE.Enc}). 
		\item \textbf{Game 1.} This game slightly changes the way $\textbf{B}_0, \cdots \textbf{B}_d$ are generated in the setup phase. 
		Instead in the challenge phase, $\widehat{\textbf{S}}_1, \cdots, \widehat{\textbf{S}}_{d} \in \{-1,1\}^{m \times m}$ are sampled in the setup phase.
		This allows to compute $\textbf{B}_i:=\textbf{A}\widehat{\textbf{S}}_i-\widehat{t_i}\textbf{G}$ for $i\in [d]$.
		The rest of the game is the same as Game 0. \\
		Game 1 and Game 0 are indistinguishable thanks to the leftover hash lemma (i.e., Lemma \ref{lhl}).
		\item  \textbf{Game 2.} In this game,  the matrix $\textbf{A}$ is not generated by \textsf{TrapGen} 
		but chosen uniformly at random from $\mathbb{Z}_q^{n \times m}$. 
		The matrices $\textbf{B}_1, \cdots \textbf{B}_d$ are constructed as in Game 1. 
		The secret key is $sk_0=\{\textbf{T}_\textbf{G}\}$ instead. 
		
		The challenger replies to a delegated key query DKQ$(y, f_1,\cdots, f_k)$ as follows:
		
		\begin{enumerate}
			\item If $f_1(\widehat{\mathbf{t}})=f_k(\widehat{\mathbf{t}})=y$, the challenger aborts and restarts the game until there exists at least one $f_j(\widehat{\mathbf{t}})\neq y$.  Without loss of generality, we can assume that $f_k(\widehat{\mathbf{t}})\neq y$.
			\item For all $i\in [k]$, compute
			$\widehat{\textbf{S}}_{f_i}\leftarrow \textsf{Eval}_\textsf{sim}(f_{i}, ((\widehat{t_j},\widehat{\textbf{S}_j}))_{j=1}^d , \textbf{A})$, and let  $\textbf{B}_{f{i}}=\textbf{A}\widehat{\textbf{S}}_{f_i}-f_{i}(\widehat{\mathbf{t}})\textbf{G}$. Remark that, $\textbf{B}_{f_{1}}=\textsf{Eval}_\textsf{pk}(f_{i}, (\textbf{B}_j)_{j=1}^d)$. For choosing Gaussian parameters, note that  $\|\widehat{\textbf{S}}_{f_i}\|_{sup} \leq \beta_{\mathcal{F}}$ due to Item 3 of Lemma \ref{eval}.
			\item $\textbf{E}_{y,f_1, \cdots, f_k} \leftarrow \textsf{ExtBasisRight}([\textbf{A}|\textbf{A}\widehat{\textbf{S}}_{f_1}+(y-f_1(\widehat{\mathbf{t}}))\textbf{G}|\cdots |\textbf{A}\widehat{\textbf{S}}_{f_k}+(y-f_k(\widehat{\mathbf{t}}))\textbf{G}], \textbf{T}_{\textbf{G}}).$ Note that,  $\| \textbf{E}_{y,f_1, \cdots, f_k}\| \leq \| \widetilde{\mathbf{T}_\mathbf{G}}\|(1+\|\mathbf{S}_{f_k}\|_{sup})=\sqrt{5}(1+\beta_{\mathcal{F}})$ for all $k \in [\eta]$ by Item 2 of Lemma \ref{trapdoor}. 
			\item $\textbf{T}_{y,f_1, \cdots, f_k} \leftarrow \textsf{RandBasis}([\textbf{A}|\textbf{A}\widehat{\textbf{S}}_{f_1}+(y-f_1(\widehat{\mathbf{t}}))\textbf{G}|\cdots |\textbf{A}\widehat{\textbf{S}}_{f_k}+(y-f_k(\widehat{\mathbf{t}}))\textbf{G}],$ $ \textbf{E}_{y,f_1, \cdots, f_k} , \sigma_k).$  
			\item Return $sk_{y,f_1, \cdots, f_k}:=\{\textbf{T}_{y, f_1, \cdots, f_k} \}$.
		\end{enumerate}			
		Game 2 and Game 1 are indistinguishable. The reason is that the distributions of $\textbf{A}$'s in both games are statistically close and that the challenger's response to the adversary's query is also the output of \textsf{RandBasis}.
		
		\item  \textbf{Game 3.} This game is similar to Game 2, 
		except that the challenge ciphertext $\widehat{ct}$ 
		is chosen randomly. 
		Therefore, the advantage of the adversary $\mathcal{A}$ in Game 3 is zero.\\
		Now we show that Games 2 and 3 are indistinguishable using a reduction from DLWE.
		
		\item \textbf{Reduction from DLWE.} Suppose that $\mathcal{A}$ can distinguish Game 2 from Game 3 
		with a non-negligible advantage. 
		Using $\mathcal{A}$, we construct a DLWE solver $\mathcal{B}$. The reduction is as follows:
		\begin{itemize}
			\item $(n,2m,q,\chi)$--\textbf{DLWE instance.} $\mathcal{B}$ is given a
			$\textbf{F} \xleftarrow{\$}\mathbb{Z}_q^{n \times 2m}$, and a vector $\mathbf{c}\in \mathbb{Z}_q^{2m}$,
			where  either (i) $\textbf{c}$ is random or 
			(ii) $\textbf{c}$ is in the LWE form
			$\textbf{c}=\textbf{F}^T \textbf{s}+\textbf{e},$
			
			for some random vector $\textbf{s}\in \mathbb{Z}_q^{n}$ and $\textbf{e} \leftarrow \chi^{2m}$. 
			The goal of $\mathcal{B}$ is to decide whether  $\textbf{c}$ is random or generated from LWE.
			
			\item \textbf{Initial.} $\mathcal{B}$ now parses   $[\mathbf{c}_{\textsf{in}}^T|\mathbf{c}_{\textsf{out}}^T]^T \leftarrow \mathbf{c}$, where $\mathbf{c}_{\textsf{in}}, \mathbf{c}_{\textsf{out}}\in \mathbb{Z}_q^{m}$, $[\mathbf{e}_{\textsf{in}}^T|\mathbf{e}_{\textsf{out}}^T]^T \leftarrow \mathbf{e}$, where $\mathbf{e}_{\textsf{in}}, \mathbf{e}_{\textsf{out}}\leftarrow \chi^{m}$, and  $[\mathbf{A} |\mathbf{U}] \leftarrow \mathbf{F}$, where $\mathbf{A}, \textbf{U}\in \mathbb{Z}_q^{n \times m}$. That is,
			\begin{equation}\label{key89}
			\textbf{c}_{\textsf{in}}=\textbf{A}^T \textbf{s}+\textbf{e}_{\textsf{in}},	 \quad \textbf{c}_{\textsf{out}}=\textbf{U}^T \textbf{s}+\textbf{e}_{\textsf{out}}.
			\end{equation}
			
			Now
			$\mathcal{B}$ calls $\mathcal{A}$ to get the target variable $\widehat{\mathbf{t}}=(\widehat{t_1}, \cdots, \widehat{t_d})$ to be challenged.

			\item \textbf{Setup.}   
			$\mathcal{B}$ generates the keys as in Game 2. 
			That is,  $\widehat{\textbf{S}}_1, \cdots, \widehat{\textbf{S}}_{d} \xleftarrow{\$} \{-1,1\}^{m \times m}$ and 
			$\textbf{B}_i:=\textbf{A}\widehat{\textbf{S}_i}-\widehat{t_i}\textbf{G}$  for $i\in [d]$. 
			Finally, $\mathcal{B}$ sends $\mathcal{A}$ the public key $pk=(\textbf{A}, \textbf{B}_1, \cdots, \textbf{B}_d, \textbf{U})$. Also, $\mathcal{B}$ 
			keeps $sk=\{\textbf{T}_\textbf{G}\}$ as the initial secret key.
			\item \textbf{Query.} Once $\mathcal{A}$ makes a  delegated key query, $\mathcal{B}$ replies as in Game 2.
			\item \textbf{Challenge.} Once $\mathcal{A}$ submits two messages $\mu_0$ and $\mu_1$, $\mathcal{B}$ chooses uniformly at random $b \xleftarrow{\$} \{0,1\}$, then computes 
			$\widehat{\textbf{c}}\leftarrow [\textbf{I}_m|\widehat{\textbf{S}}_1|\cdots|\widehat{\textbf{S}}_{d}]^T\textbf{c}_\textsf{in} \in \mathbb{Z}_q^{(d+1)m}$  and $\widehat{\textbf{c}}_{\textsf{out}} \leftarrow \textbf{c}_{\textsf{out}}+\mu_b \lceil \frac{q}{2} \rceil \in \mathbb{Z}_q$.
			\begin{itemize}
				\item  Suppose $\textbf{c}$ is generated by LWE, 
				i.e.,  $\textbf{c}_{\textsf{in}}$, $ \textbf{c}_{\textsf{out}}$ satisfy Equation \eqref{key89}. In the \textsf{DFKHE.Enc} algorithm,
				 $\textbf{H}=[\textbf{A}|\widehat{t_1}\textbf{G}+\textbf{B}_1|\cdots |\widehat{t_d} \textbf{G}+\textbf{B}_d]=[\textbf{A}|\textbf{A}\widehat{\textbf{S}}_1|\cdots|\textbf{A}\widehat{\textbf{S}}_d]$. 
				 Then 
				$$
				\widehat{\textbf{c}}= [\textbf{I}_m|\widehat{\textbf{S}}_1|\cdots|\widehat{\textbf{S}}_{d}]^T(\textbf{A}^T\textbf{s}+\textbf{e}_\textsf{in})=\textbf{H}^T\textbf{s}+\widehat{\textbf{e}}, 
				$$
				where $\widehat{\textbf{e}}=[\textbf{I}_m|\widehat{\textbf{S}}_1|\cdots|\widehat{\textbf{S}}_{d}]^T\textbf{e}_\textsf{in}$. 
				It is easy to see that $\widehat{\textbf{c}}$ is computed as in Game 2. 	Additionally, $\widehat{\textbf{c}}_{\textsf{out}} =\textbf{U}^T \textbf{s}+\widehat{\textbf{e}}_{\textsf{out}}+\mu_b \lceil \frac{q}{2} \rceil \in \mathbb{Z}_q$. Then $\widehat{ct}:=(\widehat{\textbf{c}}, \widehat{\textbf{c}}_{\textsf{out}} ) \in \mathbb{Z}_q^{(d+2)m}$ is a valid ciphertext of $\mu_b$. 
				
				\item If $\textbf{c}_{\textsf{in}}$, $ \textbf{c}_{\textsf{out}}$ are random then $\widehat{\textbf{c}}$ is random (following a standard left over hash lemma argument). And since $\widehat{\textbf{c}}_{\textsf{out}}$ is also random, $\widehat{ct}:=(\widehat{\textbf{c}}, \widehat{\textbf{c}}_{\textsf{out}})$ is random in $\mathbb{Z}_q^{(d+2)m}$ which behaves similarly to Game 3.
			\end{itemize}
			\item \textbf{Guess.} Eventually, once $\mathcal{A}$ outputs his guess of whether he is interacting with Game 2 or Game 3, $\mathcal{B}$ outputs his decision for the DLWE problem.  
			
		\end{itemize}
		We have shown that $\mathcal{B}$ can solve the $(n,2m,q,\chi)$--$\textsf{DLWE}$ instance. \qed
		
	\end{description}
\end{proof}

\noindent \textbf{Setting Parameters.} In order to choose parameters, we should take the following into consideration:
	\begin{itemize}
		\item For the hardness of DLWE, by Theorem \ref{dlwehard}, we choose $\epsilon, n, q, \chi$, where  $\chi$ is a $\chi_0$-bounded distribution, such that $q/\chi_0\geq 2^{n^\epsilon}$. We also note that, the hardness of DLWE via the traditional worst-case reduction (e.g., Lemma \ref{dlwehard}) does not help  us much in proposing concrete parameters for lattice-based cryptosystems. Instead, a more conservative methodology that has been usually used in
the literature is the so-called ``core-SVP hardness"; see \cite[Subsection 5.2.1]{ABD+20}  for a detailed reference.
		\item Setting Gaussian parameters:
			\begin{enumerate}
		\item \textit{First approach:} Without caring the security proof, for trapdoor algorithms  to work, we can set $\sigma_1=\| \widetilde{\textbf{T}_{\textbf{A}}}\|\cdot\omega(\sqrt{\log (2m)})$, with $\| \widetilde{\textbf{T}_{\textbf{A}}}\|=O(\sqrt{n\log m})$ by Item 1 of Lemma \ref{trapdoor}. Note that, in \textsf{DFKHE.KHom} we have $\| \widetilde{\textbf{T}}_{y,f_1}\| <\sigma_1 \cdot \sqrt{2m}$  by Item 5 of Lemma \ref{trapdoor}. Then, $\sigma_2=\| \widetilde{\textbf{T}}_{y,f_1}\| \cdot \omega(\sqrt{\log (3m)})=\sigma_1 \cdot \omega(\sqrt{m\log m}).$ Similarly, we can set $\sigma_k=\sigma_1\cdot(\sqrt{m\log m})^{k-1}$ for all $k \in [\eta]$.
\item \textit{Second approach:} For  the security proof to work, we have to be careful in choosing Gaussian parameters $\sigma_1, \cdots, \sigma_{\eta}$.  Indeed, we have to choose $\sigma_1=\omega(\beta_\mathcal{F}\cdot \sqrt{\log m})$. In fact, we remarked in Step 2 of \textbf{Game 2} of the  proof for Theorem \ref{varcpa} that $\|\widehat{\textbf{S}}_{f_i}\|_{sup} \leq \beta_{\mathcal{F}}$ for all $i$.   And  for a generic $k$ we still obtain $\|\widetilde{ \textbf{E}}_{y,f_1, \cdots, f_k}\| \leq \| \widetilde{\mathbf{T}_\mathbf{G}}\|(1+\|\mathbf{S}_{f_k}\|_{sup})=\sqrt{5}(1+\beta_{\mathcal{F}})$ as we just exploit $\textbf{T}_\textbf{G}$ as the secret key. Hence, $\sigma_k=\| \widetilde{\textbf{E}}_{y,f_1, \cdots, f_k}\|\cdot\omega(\sqrt{\log ((k+1)m)})=\omega(\beta_\mathcal{F}\cdot \sqrt{\log m})$ for all $k \in [\eta]$. 
\item Compared with $\sigma_k$ of the first approach, $\sigma_k$'s of the second approach are essentially smaller. Therefore, in order for both trapdoor algorithms and the security to work, we should set $\sigma_1=\omega(\beta_\mathcal{F}\cdot \sqrt{\log m})$ and choose $\beta_{\mathcal{F}} >\| \widetilde{\textbf{T}_{\textbf{A}}}\|=\sqrt{n\log m}$ and then follow the first approach in setting Gaussian parameters.  Recall that, $\beta_{\mathcal{F}}=(\frac{p^d-1}{p-1}\cdot m)^\tau \cdot 20\sqrt{m}=O((p^{d-1}m)^\tau\sqrt{m})$ by Lemma \ref{eval2}.
	\end{enumerate}
\item For the correctness: We need Condition \eqref{key} to hold, i.e., $		(\eta+1)^2\cdot \sqrt{m}\cdot \omega( (\sqrt{m\log m})^{\eta})\cdot \beta_{\mathcal{F}}^2+2<\frac{1}{4}(q/\chi_0)$.
\end{itemize}

\noindent 	 \textbf{Sizes of Keys and Ciphertext.} Recall that, throughout this work, we set $m=\Theta(n\log q)$. The public key corresponding  $d$ variables consists of $d+1$ matrices of dimension $n\times m$ over $\mathbb{Z}_q$. Then the public key size is $O((d+1)\cdot n^2 \log^2 q)$. The initial secret key is the short trapdoor matrix $\textbf{T}_\textbf{A}$ of dimension $m\times m$ generated by \textsf{TrapGen} such that $\|\textbf{T}_\textbf{A}\| \leq O(\sqrt{n \log q})$, then size is $O(n^2 \log^2 q\cdot \log( n\log q))$. The secret key after delegating $\eta$ functions  is the trapdoor matrix $\textbf{T}_{y,f_1, \cdots, f_\eta}$ of dimension $(\eta+1)m \times (\eta+1)m$ and    $\| \textbf{T}_{y,f_1, \cdots, f_\eta} \| <\sigma_\eta\cdot \sqrt{(\eta+1)m}=\beta_{\mathcal{F}}\cdot \omega((\sqrt{m\log m})^{\eta})$ with overwhelming probability by Lemma \ref{thm:Gauss}. Therefore its size is  $ (\eta+1) \cdot n \log q \cdot( O(\log(\beta_{\mathcal{F}})+\eta\cdot \log (n \log q)))$. The ciphertext is a tuple of $(d+2)$ vectors of in $\mathbb{Z}^m_q$ hence its size is $O((d+2)\cdot n\log^2q))$. 

\subsection{LWE-based PE Construction from DFKHE} \label{cons}
We define the family of equality functions $\mathcal{F}:=\{ f_{t^*}: \mathbb{Z}_q^d \rightarrow  \mathbb{Z}_q| t^* \in \mathbb{Z}_q\}$, 
where $f_{t^*}(\mathbf{t}):=eq_{t^*}(t_1)+\cdots+ eq_{t^*}(t_d)$, $\mathbf{t}=(t_1, \cdots, t_d)$,  $eq_{t^*}: \mathbb{Z}_q\rightarrow \mathbb{Z}_q $, satisfying that $\forall t\in \mathbb{Z}_q$, $eq_{t^*}(t)=1 \text{ (mod } q)$ iff $t=t^*$, 
otherwise $eq_{t^*}(t)=0 \text{ (mod } q)$.
Then  $f_{t^*}(\textbf{t})=0 \text{ (mod } q)$ iff $eq_{t^*}(t_i)=0 \text{ (mod } q)$ if  $d<q$, 
for all $i\in [d]$. 
By applying the generic framework in Section \ref{generic} to DFKHE demonstrated in Subsection \ref{dfkhe} and modifying the resulting PE, we come up with the LWE-based  \textsf{PE} construction $\Psi=\{\textsf{PE.key},$ $ \textsf{PE.enc}, $ $\textsf{PE.pun}, \textsf{PE.dec}\}$ presented below:
\begin{description}	
	\item \underline{$\textsf{PE.key}(1^\lambda )$:} 
	For the input  security parameter $\lambda$, do the following:
	\begin{enumerate}
		\item Choose $n=n(\lambda)$, $q=q(\lambda)$ prime,  and  the maximum number of tags $d=d(\lambda)$ per a ciphertext such that $d<q$.
		\item Choose  $m=\Theta{(n\log q)}$. The plaintext space is $\mathcal{M}:= \{0,1\}^m$,  $\mathcal{T}:=\mathbb{Z}_q$. Additionally, let $\chi$ be a $\chi_0$--bounded noise distribution (i.e, its support belongs to $[-\chi_0, \chi_0]$ for which the $(n,2m, q,\chi)$--DLWE is hard. Set $\sigma=\omega(\beta_\mathcal{F}\cdot \sqrt{\log m}).$

		\item Sample $(\textbf{A},\textbf{T}_\textbf{A}) \leftarrow \textsf{TrapGen}(n,m,q)$,  $\textbf{U},\textbf{B}_1, \cdots ,\textbf{B}_d \xleftarrow{\$} \mathbb{Z}_q^{n \times m}$.
		\item Output $pk=\{\textbf{A},\textbf{B}_1, \cdots \textbf{B}_d , \textbf{U}\}$ and $sk_0=\{\textbf{T}_\textbf{A}\}$.
	\end{enumerate}
	\item \underline{$\textsf{PE.enc}(\mu,  pk, \{t_1, \cdots, t_d\})$:} 
	For the input consiting of (a message $\mu$,  the public key $pk$ and ciphertext tags $(t_1, \cdots, t_d) \in \mathcal{T}^d$), 
	perform the following steps:
	\begin{enumerate}
		\item Sample $\textbf{s} \xleftarrow{\$}\mathbb{Z}_q^{n}$, 
		$\textbf{e}_\textsf{out}, \textbf{e}_\textsf{in} \leftarrow \chi^{m}$,  $\textbf{S}_1, \cdots, \textbf{S}_{d} \xleftarrow{\$} \{-1,1\}^{m \times m}$.
		\item Compute $\textbf{e}\leftarrow (\textbf{I}_m|\textbf{S}_1|\cdots|\textbf{S}_{d})^T\textbf{e}_\textsf{in}=(\textbf{e}_{ \textsf{in}}^T, \textbf{e}_{1}^T, \cdots, \textbf{e}_{d}^T)^T $.
		\item Form $\textbf{H}\leftarrow [\textbf{A}|t_1\textbf{G}+\textbf{B}_1|\cdots |t_d \textbf{G}+\textbf{B}_d]$ 
		and compute $\textbf{c}=\textbf{H}^T\textbf{s}+\textbf{e} \in \mathbb{Z}_q^{(d+1)m}$,\\
		$\textbf{c}=[\textbf{c}_\textsf{in}|\textbf{c}_1|\cdots |\textbf{c}_d]$, where $\textbf{c}_{\textsf{in}}=\textbf{A}^T \textbf{s}+\textbf{e}_{\textsf{in}}$
		and  $\textbf{c}_{i}=(t_i\textbf{G}+\textbf{B}_i)^T \textbf{s}+\textbf{e}_{i}$ for $i\in [d]$.
		\item Compute $\textbf{c}_{\textsf{out}} \leftarrow \textbf{U}^T \textbf{s}+\textbf{e}_{\textsf{out}}+\mu \lceil \frac{q}{2} \rceil$, output $(ct= (\textbf{c}_{\textsf{in}}, \textbf{c}_1, \cdots, \textbf{c}_d, \textbf{c}_{\textsf{out}}), (t_1, $ $ \cdots, t_d)) $.
	\end{enumerate}

	\item \underline{$\textsf{PE.pun}(sk_{\eta-1},t^*_\eta)$: }
	For the input (a puncture key $sk_{\eta-1}$ and a punctured tag $t^*_\eta \in \mathcal{T}$), do:
	\begin{enumerate}
		\item Evaluate $\textbf{B}_{eq_\eta} \leftarrow \textsf{Eval}_\textsf{pk}(f_{t^*_{\eta}}, (\textbf{B}_k)_{k=1}^d)$.
		\item Compute  $\textbf{E}_{eq_{\eta}} \leftarrow \textsf{ExtBasisLeft}([\textbf{A}|\textbf{B}_{eq_{1}}|\cdots |\textbf{B}_{eq_{{\eta-1}}}|\textbf{B}_{eq_{{\eta}}}], \textbf{T}_{eq_{{\eta-1}}})$.
			\item  $\textbf{T}_{eq_{\eta}} \leftarrow \textsf{RandBasis}([\textbf{A}|\textbf{B}_{eq_{1}}|\cdots |\textbf{B}_{eq_{{\eta-1}}}|\textbf{B}_{eq_{{\eta}}}], \textbf{E}_{eq_{\eta}}, \sigma_\eta)$.
		\item Output $sk_{\eta}:=(\textbf{T}_{eq_{\eta}},(t^*_1,\cdots, t^*_{\eta}), (\textbf{B}_{eq_{1}}, \cdots, \textbf{B}_{eq_{\eta}}))$.
	\end{enumerate}
	
	\item \underline{$\textsf{PE.dec}(ct, \textbf{t}, (sk_\eta, \{t^*_1, \cdots, t^*_{\eta}\}))$: }
	For the input (a ciphertext $ct=(\textbf{c}_{\textsf{in}}, \textbf{c}_1, \cdots, \textbf{c}_{d},$ $ \textbf{c}_{\textsf{out}}) $, 
	the associated tags $\textbf{t}=(t_1, \cdots, t_d)$, 
	a puncture key $sk_{\eta}$ and the associated punctured tags $\{t^*_1, \cdots, t^*_{\eta}\} \subset \mathcal{T}$), execute the following steps:
	\begin{enumerate}
		\item If there exists $j\in [\eta]$ such that $f_{t^*_j}(\textbf{t}) \neq 0$, then output $\bot$. Otherwise, go to Step 2.
		\item Parse $sk_{\eta}:=(\textbf{T}_{eq_{\eta}},(t^*_1,\cdots, t^*_{\eta}), (\textbf{B}_{eq_{1}}, \cdots, \textbf{B}_{eq_{\eta}}))$.
		\item Sample $\textbf{R} \leftarrow \textsf{SampleD}([\textbf{A} |\textbf{B}_{eq_{1}}|\cdots |\textbf{B}_{eq_{{\eta}}}], \textbf{T}_{eq_{{\eta}}}, \textbf{U}, \sigma_\eta)$.
		\item Evaluate  $\textbf{c}_{eq_j}\leftarrow \textsf{Eval}_\textsf{ct}(f_{t^*_j},  ((t_k, \textbf{B}_k, \textbf{c}_{k}))_{k=1}^{d})$, for $j\in [\eta]$.
		\item Compute $\bar{\mu}=(\bar{\mu}_1, \cdots, \bar{\mu}_m) \leftarrow \textbf{c}_{\textsf{out}}-\textbf{R}^T(\textbf{c}_\textsf{in}|\textbf{c}_{eq_{1}}|\cdots|\textbf{c}_{eq_\eta})$.
		
		\item For $\ell \in [m]$,  if $ |\bar{\mu}_\ell | <q/4$ then output $\mu_\ell=0$;  otherwise, output $\mu_\ell=1$.
	\end{enumerate}  
\end{description}

We remark that all analysis done for the LWE-based DFKHE in Subsection \ref{dfkhe} can perfectly applied to our LWE-based PE. Therefore, we do not mention the analysis again in this section. For completeness, we only state two main theorems as below.
\begin{theorem}[Correctness of $\Psi$] 
	The proposed $\mathsf{PE}$ $\Psi$ is correct 
	if 
			$(\eta+1)^2\cdot m^{1+\frac{\eta}{2}}\cdot \omega( (\sqrt{\log m})^{\eta+1})\cdot \beta_{\mathcal{F}}^2+2<\frac{1}{4}(q/\chi_0),$
 assumming   that $t^*_j \neq t_k$ for all $(j,k)\in [\eta]\times [d]$.
\end{theorem}
\begin{theorem}[IND-sPUN-CPA] 
	The proposed PE $\Psi$ scheme is IND-sPUN-CPA thanks to the IND-sVAR-CPA of the underlying DFKHE $\Pi$.
\end{theorem}

\section{Discussion on Unbounded Number of Ciphertext Tags} \label{unbounded} 
The  idea of   \cite{BV16}   might help us to extend the LWE-based DFKHE construction from Subsection \ref{dfkhe} 
(resp.,  PE from Subsection \ref{cons}) to a variant   
that supports arbitrary number of variables (resp., ciphertext tags). We call this variant \textit{\textsf{unDFKHE}}. 
Although, the original idea of \cite{BV16} is applied to ABE with attributes belonging to $\{0,1\}$ using the XOR operation, we believe that it might be adapted to work well with our DFKHE with variables and punctures over $\mathbb{Z}_q$ using the addition modulo $q$ (denoted $\oplus_q$. 

In \textsf{unDFKHE}, the maximum number of ciphertext tags $d$ is not fixed  in advance. 
Then, in the key generation algorithm, we cannot generate $\textbf{B}_1, \cdots, \textbf{B}_d$ and give them to the public. 
In order to solve this issue, we utilize a  family of pseudorandom functions \textsf{PRF}=(\textsf{PRF.Gen}, \textsf{PRF.Eval}), 
where $\textsf{PRF.Gen}(1^\lambda)$ takes as input a security parameter $\lambda$ and outputs a seed 
$\textbf{s} \in \mathbb{Z}_q^\ell$ of length $\ell=\ell({\lambda})$ (which depends on  $\lambda$) and $\textsf{PRF.Eval}(\textbf{s},\textbf{x})$ 
takes as input a seed $\textbf{s} \in \mathbb{Z}_q^{\ell}$ and a variable $\textbf{x}\in \mathbb{Z}_q^*$ of \textit{arbitrary length} and  
returns an element in $\mathbb{Z}_q$. 
The family of pseudorandom functions  helps us to stretch a variable of fixed length $\ell$ to one of arbitrary length $d $ as follows. In \textsf{unDFKHE.KGen}, for a variable $\textbf{t}$ of length $d=|\textbf{t}|$, instead of  $\textbf{B}_1, \cdots, \textbf{B}_d$, we generate $\overline{\textbf{B}}_1, \cdots, \overline{\textbf{B}}_\ell$ and use them to produce $\textbf{B}_1, \cdots, \textbf{B}_d$ later. This can be done by running $\textsf{Eval}_\textsf{pk}(\textsf{PRF.Eval}(\cdot,i), (\overline{\textbf{B}}_k)_{k=1}^\ell)$, for $i\in [d]$, where $\textsf{PRF.Eval}(\cdot,i)$ acts as a function that can be evaluated by $\textsf{Eval}_\textsf{pk}$. Accordingly, any function $f \in \mathcal{F}$ will also be transformed to $f_{\Delta}$ defined by $f_{\Delta}(\textbf{t}):=f(\textbf{t}\oplus_q\Delta_{\le d})$ before joining to any computation later on. Here $\Delta_i:=\textsf{PRF.Eval}(\textbf{s},i)$  for $i\in [d]$, $\Delta_{\le d}=(\Delta_1,\cdots, \Delta_d)$.  Also remark that, $f_{\Delta}(\textbf{t}\oplus_q (q_{\le d}-\Delta_{\le d}))=f(\textbf{t})$, where $q_{\le d}=(q, \cdots, q) \in \mathbb{Z}^d$.
Therefore, in \textsf{unDFKHE.KHom},  $\textbf{B}_{f} \leftarrow \textsf{Eval}_\textsf{pk}(f_{\Delta}, (\textbf{B}_k)_{k=1}^d)$.

Actually, there are a lot of work left to be done. Due to space limitation, we leave details of this section for the full version of this paper.

\section{Conclusion and Future Works} \label{conclude}
In this paper, we show puncturable encryption can be constructed from the so-called delegatable fully key-homomorphic encryption.
From the framework, we instantiate our puncturable encryption construction using LWE. 
Our puncturable encryption enjoys the selective indistinguishability under chosen plaintext attacks, 
which can be converted into adaptive indistinguishability under chosen ciphertext attacks using well-known standard techniques.
For future works, there are few investigation directions worth pursuing such as design of:
(i) puncturable lattice-based ABE as in \cite{PNXW18},
(ii) efficient puncturable forward-secure encryption schemes as proposed in\cite{GM15} or 
(iii) puncturable encryption schemes, whose  puncture key size is constant or puncturable ecnryption schemes support unlimited number of punctures.

\subsubsection{Acknowledgment.} 
We thank Sherman S.M. Chow and anonymous reviewers for their insightful comments which improve the content and presentation of the manuscript.  This work is supported by the Australian Research Council Linkage Project LP190100984. Huy Quoc Le has been sponsored by a CSIRO Data61 PhD Scholarship and CSIRO Data61 Top-up Scholarship.  Josef Pieprzyk has been supported by the Australian ARC grant DP180102199 and Polish NCN grant 2018/31/B/ST6/03003.

\end{document}